\documentclass[12pt]{article}
\usepackage{fullpage}
\usepackage{mathrsfs}
\usepackage{float}
\usepackage{amsmath, amsfonts, amssymb, amsthm, amsbsy, amscd, bm, bbm}
\usepackage{array}
\usepackage{booktabs}
\usepackage{graphicx}
\usepackage[small,bf]{caption}
\usepackage{subcaption}
\usepackage{color}
\usepackage{ifthen}
\usepackage{xspace}
\usepackage{algorithm, algorithmic}
\usepackage{hyperref}
\usepackage{url}
\usepackage{tocbibind}
\usepackage{comment}
\usepackage{natbib}

\DeclareMathOperator*{\argmin}{argmin}

\newtheorem{theorem}{Theorem}[section]
\newtheorem{corollary}[theorem]{Corollary}

\newtheorem{definition}[theorem]{Definition}
\newtheorem{lemma}[theorem]{Lemma}
\newtheorem{remark}{Remark}[section]

\newtheorem{fact}[theorem]{Fact}

\newcommand{\supp}{\mathrm{supp}\,}
\newcommand{\eps}{\varepsilon}
\newcommand{\var}{\mathrm{Var}}
\newcommand{\norm}[1]{\lVert #1\rVert}

\def\sgn{{\mathrm{sign}}}
\def\E {{\mathbb{E}}}
\def\eps{{\varepsilon}}
\def\bp{{\mathbf{p}}}
\def\bP{{\mathbf{P}}}

\def\bq{{\mathbf{q}}}
\def\Ber{{\mathcal{B}}}
\def\cP{{\mathcal{P}}}
\def\hP{{\hat{P}}}
\def\I{{\mathbb{I}}}
\def\A{{\mathcal{A}}}

\newcommand{\domination}{\textsc{Domination}}
\newcommand{\topk}{\textsc{Top-K}}

\title{Competitive analysis of the top-$K$ ranking problem}
\author{
Xi Chen \thanks{Stern School of Business, New York University, email: xchen3@stern.nyu.edu}
\and
Sivakanth Gopi \thanks{Department of Computer Science, Princeton University, email: sgopi@cs.princeton.edu}
\and
Jieming Mao  \thanks{Department of Computer Science, Princeton University, email: jiemingm@cs.princeton.edu}
\and
Jon Schneider \thanks{Department of Computer Science, Princeton University, email: js44@cs.princeton.edu}
}

\begin{document}

\maketitle

\begin{abstract}
Motivated by applications in recommender systems, web search, social choice and crowdsourcing, we consider the problem of identifying the set of top $K$ items from noisy pairwise comparisons. In our setting, we are non-actively given $r$ pairwise comparisons between each pair of $n$ items, where each comparison has noise constrained by a very general noise model called the strong stochastic transitivity (SST) model. We analyze the competitive ratio of algorithms for the top-$K$ problem. In particular, we present a linear time algorithm for the top-$K$ problem which has a competitive ratio of $\tilde{O}(\sqrt{n})$; i.e. to solve any instance of top-$K$, our algorithm needs at most $\tilde{O}(\sqrt{n})$ times as many samples needed as the best possible algorithm for that instance (in contrast, all previous known algorithms for the top-$K$ problem have competitive ratios of $\tilde{\Omega}(n)$ or worse). We further show that this is tight: any algorithm for the top-$K$ problem has competitive ratio at least $\tilde{\Omega}(\sqrt{n})$. 
\end{abstract}
\newpage
\tableofcontents
\newpage
\section{Introduction}\label{sec:intro}

%!TEX root =  main.tex

The problem of inferring a ranking over a set of $n$ items, such as documents, images, movies, or URL links, is an important problem in machine learning and finds many applications in recommender systems, web search, social choice, and many other areas. One of the most popular forms of data for ranking is pairwise comparison data, which can be easily collected via, for example, crowdsourcing, online games, or tournament play. The problem of ranking aggregation from pairwise comparisons has been widely studied and most work aims at inferring a total ordering of all the items (see, e.g., ~\cite{Negahban12RankCentrality}). However, for some applications with a large number of items (e.g., rating of restaurants in a city), it is only necessary to identify the set of top $K$ items. For these applications, inferring the total global ranking order unnecessarily increases the complexity of the problem and requires significantly more samples.

In the basic setting for this problem, there is a set of $n$ items with some true underlying ranking. For possible pair $(i, j)$ of items, an analyst is given $r$ noisy pairwise comparisons between those two items, each independently ranking $i$ above $j$ with some probability $p_{ij}$. From this data, the analyst wishes to identify the top $K$ items in the ranking, ideally using as few samples $r$ as is necessary to be correct with sufficiently high probability. The noise in the pairwise comparisons (i.e. the probabilities $p_{ij}$) is constrained by the choice of noise model. Many existing models - such as the Bradley-Terry-Luce model (BTL) \citep{Bradley52, Luce59}, the Thurstone model \citep{Thurstone:27}, and their variants - are \emph{parametric} comparison models, in that each probability $p_{ij}$ is of the form $f(s_i, s_j)$, where $s_i$ is a `score' associated with item $i$. While these parametric models yield many interesting algorithms with provable guarantees \citep{Chen15, Jang13, Suh16Adversarial}, the models enforce strong assumptions on the probabilities of incorrect pairwise comparisons that might not hold in practice \citep{Davidson59,McLaughlin65, Tversky72, Ballinger97}.

A more general class of pairwise comparison model is the strong stochastic transitivity (SST) model, which subsumes the aforementioned parameter models as special cases and has a wide range of applications in psychology and social science (see, e.g., \cite{Davidson59,McLaughlin65,Fishburn73}). The SST model only enforces the following coherence assumption:  if $i$ is ranked above $j$, then $p_{il} \ge p_{jl}$ for all other items $l$.  \cite{Shah15Sto} pioneered the algorithmic and theoretical study  of ranking aggregation under SST models. For top-$K$ ranking problems, \cite{Shah15Sim} proposed a counting-based algorithm, which simply orders the items by the total number of pairwise comparisons won. For a certain class of instances, this algorithm is in fact optimal; any algorithm with a constant probability of success on these instances needs roughly at least as many samples as this counting algorithm. However, this does not rule out the existence of other instances where the counting algorithm performs asymptotically worse than some other algorithm.

In this paper, we study algorithms for the top-$K$ problem from the standpoint of \emph{competitive analysis}. We give an algorithm which, on any instance, needs at most $\tilde{O}(\sqrt{n})$ times as many samples as the best possible algorithm for that instance to succeed with the same probability. We further show this result is tight: for any algorithm, there are instances where that algorithm needs at least $\tilde{\Omega}(\sqrt{n})$ times as many samples as the best possible algorithm. In contrast, the counting algorithm of \cite{Shah15Sto} sometimes requires $\Omega(n)$ times as many samples as the best possible algorithm, even when the probabilities $p_{ij}$ are bounded away from $1$.
 
Our main technical tool is the introduction of a new decision problem we call \emph{domination}, which captures the difficulty of solving the top $K$ problem while being simpler to directly analyze via information theoretic techniques. The domination problem can be thought of as a restricted one-dimensional variant of the top-$K$ problem, where the analyst is only given the outcomes of pairwise comparisons that involve item $i$ or $j$, and wishes to determine whether $i$ is ranked above $j$. Our proof of the above claims proceeds by proving analogous competitive ratio results for the domination problem, and then carefully embedding the domination problem as part of the top-$K$ problem.

\subsection{Related Work}

The problem of sorting a set of items from a collection of pairwise comparisons is one of the most classical problems in computer science and statistics. Many works investigate the problem of recovering the total ordering under noisy comparisons drawn from some parametric model. For the BTL model, Negahban et al. \cite{Negahban12RankCentrality} propose the \emph{RankCentrality} algorithm, which serves as the building block for many spectral ranking algorithms. Lu and Boutilier \cite{Craig:11} give an algorithm for sorting in the Mallows model. Rajkumar and Agarwal \cite{Rajkumar14} investigate which statistical assumptions (BTL models, generalized low-noise condition, etc.) guarantee convergence of different algorithms to the true ranking.

More recently, the problem of top-$K$ ranking has received a lot of attention. Chen and Suh \cite{Chen15}, Jang et al. \cite{Jang13},  and Suh et al. \cite{Suh16Adversarial} all propose various spectral methods for the BTL model or a mixture of BTL models. Eriksson \cite{Eriksson13} considers a noisy observation model where comparisons deviating from the true ordering are i.i.d. with bounded probability. In  \cite{Shah15Sim}, Shah and Wainwright consider the general SST models and propose the counting-based algorithm, which motivates our work. The top-$K$ ranking problem is also related to the best $K$ arm identification in multi-armed bandit \cite{Bubeck:13,Jamieson:14,Zhou:14}. However, in the latter problem, the samples are i.i.d. random variables rather than pairwise comparisons and the goal is to identify the top $K$ distributions with largest means.

This paper and the above references all belong to the \emph{non-active} setting: the set of data provided to the algorithm is fixed, and there is no way for the algorithm to adaptively choose additional pairwise comparisons to query. In several applications, this property is desirable, specifically if one is using a well-established dataset or if adaptivity is costly (e.g. on some crowdsourcing platforms). Nonetheless, the problems of sorting and top-$K$ ranking are incredibly interesting in the adaptive setting as well. Several works \citep{Ailon11, Jamieson11, Mathieu07, Braverman08} consider the adaptive noisy sorting problem with (noisy) pairwise comparisons and explore the sample complexity to recover an (approximately) correct total ordering in terms of some distance function (e.g,., Kendall's tau). In \cite{Wauthier13}, Wauthier et al. propose simple weighted counting algorithms to recovery an approximate total ordering from noisy pairwise comparisons. Dwork et al. \cite{Dwork01} and Ailon et al. \cite{Ailon08} consider a related \emph{Kemeny optimization} problem, where the goal is  to determine the total ordering that minimizes the sum of the distances to different permutations. More recently, the top-$K$ ranking problem in the active setting has been studied by Braverman et al. \cite{BMW16} where they consider the sample complexity of algorithms that use a constant number of rounds of adaptivity. All of this work takes place in much more constrained noise models than the SST model. Extending our work to the active setting is an interesting open problem.

%\jiemingnote{Jieming: I am not sure if we should put the following paragraph here.}

\section{Preliminaries and Problem Setup}\label{sec:prelim}
%!TEX root =  main.tex

%\subsection{Problem Setup}

%\jonnote{Probably a lot of the formal definitions here should be moved to the Preliminaries section and should be replaced by informal discussion (or we should state our main results informally in the previous section).}
%\jiemingnote{Jieming: I am bit worried that it's hard to discuss many things before the formal definitions are made. The informal discussions might be very confusing. As long as we don't want to put the definitions in the appendix, I think it's okay to leave them in the introduction. The definition of \domination~can be moved to the main technique section.}

Consider the following problem. An analyst is given a collection of $n$ items, labelled $1$ through $n$. These items have some true ordering defined by a permutation $\pi: \{1,\ldots, n\} \rightarrow \{1, \ldots, n\}$ such that for $1 \leq u<v \leq n$, the item labelled $\pi(u)$ has a better rank than the item labelled $\pi(v)$ (i.e., the item with label $i$ has a better rank than the item $j$ if and only if $\pi^{-1}(i) <  \pi^{-1}(j)$). The analyst's goal is to determine the set of the top $K$ items, i.e., $\{\pi(1), \ldots, \pi(k)\}$.

The analyst receives $r$ samples. Each sample consists of pairwise comparisons between all pairs of items. All the pairwise comparisons are independent with each other. The outcomes of the pairwise comparison between any two items is characterized by the probability matrix  $\bP \in [0,1]^{n \times n}$.  For a pair of items $(i,j)$,  let $X_{i,j} \in \{0,1\}$ be the outcome of the comparison between the item $i$ and $j$, where $X_{i,j}=1$ means $i$ is preferred to $j$ (denoted by $i \succ j$) and $X_{i,j}=0$ otherwise. Further, let $\Ber(z)$ denote the Bernoulli random variable with  mean $z \in [0,1]$. The outcome $X_{i,j}$ follows $\Ber(\bP_{\pi^{-1}(i), \pi^{-1}(j)})$, i.e.,
\[
\Pr(X_{i,j}=1)=\Pr(i \succ j)= \bP_{\pi^{-1}(i), \pi^{-1}(j)}.
\]
The probability matrix $\bP$ is said to be strong stochastic transitive (SST) if it satisfies the following definition.
\begin{definition}\label{def:SST}
The $n\times n$ probability matrix $\bP \in [0,1]^{n \times n}$ is strong stochastic transitive (SST) if
\begin{enumerate}
  \item For $1 \leq u < v \leq n$, $\bP_{u,l} \geq \bP_{v,l}$ for all $l \in [n]$.
  \item $\bP$ is shifted-skew-symmetric (i.e., $\bP-0.5$ is skew-symmetric) where $\bP_{v,u} =1-\bP_{u,v}$ and $\bP_{u,u}=0.5$ for $u\in [n]$.
\end{enumerate}
\end{definition}
The first condition claims that when the item $i$ has a higher rank than item $j$ (i.e., $\pi^{-1}(i) < \pi^{-1}(j)$), for any other item $k$, we have
\[
\Pr(i \succ k) = \bP_{\pi^{-1}(i),\pi^{-1}(k)} \geq \Pr(j \succ k) = \bP_{\pi^{-1}(j),\pi^{-1}(k)}.
\]

\begin{remark}
Many classical parametric models such that BTL  \citep{Bradley52, Luce59}  and Thurstone (Case V) \citep{Thurstone:27}  models are special cases of SST. More specifically, parametric models assume a score vector $w_1 \geq w_2 \geq \ldots \geq w_n$. They further assume that  the comparison probability $\bP_{u,v}= F(w_u-w_v)$, where $F: \mathbb{R} \rightarrow [0,1]$ is a non-decreasing function and $F(t)=1-F(-t)$ (e.g., $F(t)=1/(1+\exp(-t))$ in BTL models). By the property of $F$, it is easy to verify that $\bP_{u,v}= F(w_u-w_v)$ satisfy the conditions in Definition \ref{def:SST}.
\end{remark}

Under the SST models, we can formally define the top-$K$ ranking problem as follows. The top-$K$ ranking problem takes the inputs $n$, $k$, $r$ that are known to the algorithm and the SST probability matrix $\bP$ that is unknown to the algorithm.
\begin{definition}\label{def:topK}
\topk$(n,k,\bP,r)$ is the following algorithmic problem:
\begin{enumerate}
\item A permutation $\pi$ of $[n]$ is uniformly sampled.
\item The algorithm is given samples $X_{i,j,l}$ for $i \in [n], j \in [n], l\in [r]$, where each $X_{i,j,l}$ is sampled independently according to $\Ber(\bP_{\pi^{-1}(i), \pi^{-1}(j)})$. The algorithm is also given the value of $k$, but not $\pi$ or the matrix $\bP$.
\item The algorithm succeeds if it correctly outputs the set of labels $\{ \pi(1),...,\pi(k) \}$ of the top $k$ items.
\end{enumerate}
\end{definition}

\begin{remark}We note that \cite{Shah15Sim} considers a slightly different observation model in which each pair is queried $r$ times. For each query, one can obtain a comparison result with the probability  $p_{\text{obs}} \in (0,1]$ and with probability $1-p_{\text{obs}}$, the query is invalid. In this model, each pair will be compared $r \cdot  p_{\text{obs}}$ times on expectation. When $p_{\text{obs}}=1$, it reduces to our model in Definition \ref{def:topK}, where we observe exactly $r$ comparisons for each pair. Our results can be easily extended to deal with the observation model in \cite{Shah15Sim} by replacing $r$ with the effective sample size, $r \cdot  p_{\text{obs}}$. We  omit the details for the sake of simplicity.  %\cite{Chen15} assumed a Erd\"{o}s-R\'{e}nyi graph based observation model each each pair is chosen to be compared with certain probability. It will be interesting to generalize our results to such a comparison model in the future.
%\xnote{Xi: please add notes if necessary how we can extend to the observation models in \cite{Shah15Sim} and \cite{Chen15}}
%\jiemingnote{Jieming: There's a simple argument to extend to the observation model by only losing a $\log(n)$ factor. As we don't care about polylogarithmic factors in this paper, it's okay. And in fact, the observation model is kind of weird. We don't have to follow their model.}
\end{remark}

Our primary metric of concern is the \textit{sample complexity} of various algorithms; that is, the minimum number of samples an algorithm $A$ requires to succeed with a given probability. To this end, we call the triple $S = (n, k, \bP)$ an \textit{instance} of the $\topk$ problem, and write $r_{min}(S, A, p)$ to denote the minimum value such that for all $r \geq r_{min}(S, A, p)$, $A$ succeeds on instance $S$ with probability $p$ when given $r$ samples. When $p$ is omitted, we will take $p=\frac{3}{4}$; i.e., $r_{min}(S, A) = r_{min}(S, A, \frac{3}{4})$.

Instead of working directly with $\topk$, we will spend most of our time working with a problem we call $\domination$, which captures the core of the difficulty of the $\topk$ problem. $\domination$ is formally defined as follows.

\begin{definition}
\label{def:dom}
\domination$(n, \bp, \bq, r)$ is the following algorithmic problem:
\begin{enumerate}
\item  $\bp =(p_1,\cdots,p_n)$ and $\bq=(q_1,\cdots,q_n)$ are two vectors of probabilities that satisfy $1 \geq p_i \geq q_i \geq 0$ for all $i\in [n]$. $\bp,\bq$ are not given to the algorithm.
\item A random bit $B$ is sampled from $\Ber(\frac{1}{2})$. Samples $X_{i,j},Y_{i,j}$ (for $i\in [n], j\in [r]$) are generated as follows:
\begin{enumerate}
\item Case $B=0$: each $X_{i,j}$ is independently sampled according to $\Ber(p_i)$ and each $Y_{i,j}$ is independently sampled according to $\Ber(q_i)$.
\item Case $B=1$: each $X_{i,j}$ is independently sampled according to $\Ber(q_i)$ and each $Y_{i,j}$ is independently sampled according to $\Ber(p_i)$.
\end{enumerate}
The algorithm is given the samples $X_{i,j}$ and $Y_{i,j}$, but is not given the bit $B$ or the values of $\bp$ and $\bq$.
\item
The algorithm succeeds if it correctly outputs the value of the hidden bit $B$.
\end{enumerate}
\end{definition}

As before, we are interested in the sample complexity of algorithms for \domination. We call the triple $C = (n, \bp, \bq)$ an instance of $\domination$, and write $r_{min}(C, A, p)$ to be the minimum value such that for all $r\geq r_{min}(C, A, p)$, $A$ succeeds at solving $\domination(n, \bp, \bq, r)$ with probability at least $p$ (similarly, we let $r_{min}(C, A) = r_{min}(C, A, \frac{3}{4})$).

%In particular, let $Y_{i,j} \in \{0,1\}$ be the outcome of the comparison between the item $i$ and $j$, where $Y_{i,j}=1$ means $i$ is preferred to $j$ (denoted by $i \succ j$) and $Y_{i,j}=0$ otherwise. The outcome $Y_{i,j}$ follows a Bernoulli distribution with the mean $\bP_{i,j}$, i.e., $\Pr[Y_{i,j}=1]=\bP_{i,j}$. Note that since $Y_{j,i}=1-Y_{i,j}$, the probability matrix $\bP$ is shifted-skew-symmetric where $\bP_{j,i} =1-\bP_{i,j}$ and $\bP_{i,i}=0.5$. The SST model enforces that for any pair $i$ and $j$:
%\begin{equation}
%  \pi^*(i) > \pi^*(j)   \; \Leftrightarrow  \; \bP_{i,k} \geq  \bP_{j,k}, \;  \forall k \in \{1,\ldots, n\}.
%\end{equation}
%In other words, if the item $i$ has a higher rank than $j$, we have $\Pr(i \succ k) \geq \Pr(j \succ k)$ for all $k$. It is easy to observe that
%

\section{Main Results}
\label{sec:result}

There are at least two main approaches one can take to analyze the sample complexity of problems like $\topk$ and $\domination$. The first (and more common) is to bound the  value of $r_{min}(S, A)$ by some explicit function $f(S)$ of the instance $S$. This is the approach taken by \cite{Shah15Sim}. They show that for some simple function $f$ (roughly, the square of the reciprocal of the absolute difference of the sums of the $k$-th and $(k+1)$-th rows of the matrix $\bP$ i.e. $1/\norm{\bP_{k}-\bP_{k+1}}_1^2$), there is an algorithm $A$ such that for all instances $S$, $r_{min}(S, A) = O(f(S))$; moreover this is optimal in the sense that there exists an instance $S$ such that for all algorithms $A$, $r_{min}(S, A) = \Omega(f(S))$. While this is a natural approach, it leaves open the question of what the correct choice of $f$ should be; indeed, different choices of $f$ give rise to different `optimal' algorithms $A$ which outperform each other on different instances.

In this paper, we take the second approach, which is to compare the sample complexity of an algorithm on an instance to the sample complexity of the best possible algorithm on that instance. Formally, let $r_{min}(S, p) = \inf_{A} r_{min}(S, A, p)$ and let $r_{min}(S) = r_{min}(S, \frac{3}{4})$. An ideal algorithm $A$ would satisfy $r_{min}(S, A) = \Theta(r_{min}(S))$ for all instances $S$ of $\topk$; more generally, we are interested in bounding the ratio between $r_{min}(S, A)$ and $r_{min}(S)$. We call this ratio the \emph{competitive ratio} of the algorithm, and say that an algorithm is $f(n)$-competitive if $r_{min}(S, A) \leq f(n)r_{min}(S)$. (We likewise define all the corresponding notions for $\domination$).

%\jonnote{maybe include theorem statements here?}
In our main upper bound result, we give a linear-time algorithm for $\topk$ which is $\tilde{O}(\sqrt{n})$-competitive  (restatement of Corollary \ref{cor:ubmain}):
\begin{theorem}
\label{thm:ubintro}
There is an algorithm $A$ for \topk\ such that $A$ runs in time  $O(n^2r)$ and on every instance $S$ of \topk\ on $n$ items,
\[
r_{min}(S, A) \leq O(\sqrt{n}\log n)r_{min}(S).
\]
\end{theorem}
In our main lower bound result, we show that up to logarithmic factors, this $\sqrt{n}$ competitive ratio is optimal (restatement of Theorem \ref{thm:sstlb}):
\begin{theorem}
\label{thm:lbintro}
For any algorithm $A$ for \topk, there exists an instance $S$ of \topk\ on $n$ items such that
\[
r_{min}(S,A) \geq \Omega\left(\frac{\sqrt{n}}{\log n}\right) r_{min}(S).
\]
\end{theorem}

In comparison, for the counting algorithm $A'$ of \cite{Shah15Sim}, there exist instances $S$ such that $r_{min}(S, A') \geq \tilde{\Omega}(n)r_{min}(S)$. For example, consider the instance $S = (n,k,\bP)$ with

\[
\bP = \begin{bmatrix}
\frac{1}{2}&\frac{1}{2}+\varepsilon&\cdots&\cdots &\frac{1}{2}+\varepsilon\\
\frac{1}{2}-\varepsilon & & &  & \vdots \\
\vdots& \ddots&\ddots&\ddots & \vdots\\
\vdots & & & & \frac{1}{2}+\varepsilon \\
\frac{1}{2}-\varepsilon & \cdots & \cdots &\frac{1}{2}-\varepsilon & \frac{1}{2}
\end{bmatrix}
\]

It is straightforward to show that with $\Theta(\log n/ \varepsilon^2)$ samples, we can learn all pairwise comparisons correctly with high probability by taking a majority vote, and therefore even sort all the elements correctly. This implies that $r_{min}(S) =  O(\log n/ \varepsilon^2)$. On the other hand, we show in Corollary \ref{cor:countingfails} that $r_{min}(S,A') =\Omega(n/ \varepsilon^2)$ when $\varepsilon < 1/10$.
%\jonnote{say something about domination}

%We briefly summarize our main results as follows. Under the oracle case that $\bP$ is known to the algorithm, let $r_{\min}$ be the minimum sample complexity for any algorithm to solve \topk$(n,k,\bP,r)$  with high probability (w.h.p.).  Using $r_{\min}$ as the benchmark, our main results include the following upper and lower bound results.
%\begin{enumerate}
%\item Upper bound result:  Theorem \ref{thm:algmain} together with Lemma \ref{lem:topk_lb} show that we have an algorithm uses at most $\tilde{O}(\sqrt{n} \cdot r_{min})$ samples to solve \topk$(n,k,\bP,r)$ w.h.p. on any SST instance.
%\item Lower bound result: Theorem \ref{thm:sstlb} shows that for any algorithm, there exists a SST instance that it needs at least $\tilde{\Omega}(\sqrt{n} \cdot r_{min})$ samples to solve \topk$(n,k,\bP,r)$ w.h.p.
%\end{enumerate}

%Our key observation is that \topk$(n,k,\bP,r)$ can be reduced to the following \emph{domination} problem.

%\xnote{Xi: more comparisons to the counting procedure in \cite{Shah15Sim}.}

%We can start another subsection called "main techniques" to discuss about our results in \domination ~and state that there are reductions between \domination~ and \topk.

\subsection{Main Techniques and Overview}
We prove our main results by first proving similar results for \domination~which we defined in Definition~\ref{def:dom}. Intuitively $\domination$ captures the main hardness of \topk\ while being much simpler to analyze. Once we prove upper bound and lower bounds for the sample complexity of \domination, we will use reductions to prove analogous results for \topk.

We begin in Section~\ref{sec:domlb}, by proving a general lower bound on the sample complexity of domination. Explicitly, for a given instance $C = (n, \bp, \bq)$ of $\domination$, we show that $r_{min}(C)\ge \Omega(1/\I(\bp,\bq))$ where $\I(\bp,\bq)$ is the amount of information we can learn about the bit $B$ from one sample of pairwise comparison in each of the coordinates. 

In Section \ref{sec:d2}, we proceed to design algorithms for $\domination$ restricted to instances $C = (n,\bp,\bq)$ where $\delta \leq p_i,q_i \leq 1- \delta$ for some constant $0<\delta\le 1/2$. In this regime $\I(\bp,\bq)=\Theta(1/\norm{\bp-\bq}_2^2)$, which makes it easier to argue our algorithms are not too bad compared with the optimal one. We first consider an algorithm we call the counting algorithm $\A_{count}$ (Algorithm~\ref{alg:counting}), which is a \domination~analogue of the counting algorithm proposed by \cite{Shah15Sim}. We show that $\A_{count}$ has a competitive ratio of $\tilde{\Theta}(n)$. % i.e. for any $C$ in the special regime, $r_{min}(C,A_1) = \tilde{O}(n)\cdot r_{min}(C)$ and there exists an instance $C$ in the regime that  $r_{min}(C,A_1) = \tilde{\Omega}(n)\cdot r_{min}(C)$. 
Intuitively, the main reason $\A_{count}$ fails is that $\A_{count}$ tries to consider samples from different coordinates equally important even when they are sampled from a very unbalanced distribution (for example, $p_1 \neq q_1, p_2=q_2,...,p_n=q_n$). We then consider another algorithm we call the max algorithm $\A_{max}$ (Algorithm~\ref{alg:max}) which simply finds $i'= \max_i |\sum_{j=1}^r (X_{i,j}-Y_{i,j})|$ and outputs $B$ according the sign of $\sum_{j=1}^r (X_{i',j}-Y_{i',j})$. We show $\A_{max}$ also has a competitive ratio of $\tilde{\Theta}(n)$. Interestingly, $\A_{max}$ fails for a different reason from $\A_{count}$, namely that $\A_{max}$ does not use the information fully from all coordinates when the samples are sampled from a very balanced distribution. In fact, $\A_{count}$ performs well whenever $\A_{max}$ fails and vice versa. We therefore show how combine $\A_{count}$ and $\A_{max}$ in two different ways to get two new algorithms: $\A_{comb}$ (Algorithm~\ref{alg:combining}) and $\A_{cube}$ (Algorithm~\ref{alg:sumofcubes}). We show that both of these new algorithms have a competitive ratio of $\tilde{O}(\sqrt{n})$, which is tight by Theorem~\ref{thm:domlb}.

In Section \ref{sec:domination_general}, we design algorithms for \domination~in the general regime. In this regime, $\I(\bp,\bq)$ can be much larger than $\norm{\bp-\bq}_2^2$, particularly for values of $p_i$ and $q_i$ very close to $0$ or $1$. In these corner cases, the counting algorithm $\A_{count}$ and max algorithm $\A_{max}$ can fail very badly; we will show that even for fixed $n$, their competitive ratios can grow arbitrarily large (Lemma~\ref{lem:countingfails2} and Lemma~\ref{lem:maxfails2}). One main reason for this failure is that, even when $|p_i -q_i| < |p_j-q_j|$, samples from coordinate $i$ could convey much more information than the samples from coordinate $j$ (consider, for example, $p_i =\varepsilon/2, q_i =0$, and $p_j= 1/2+\varepsilon, q_j =1/2$). Taking this into account, we design a new algorithm $\A_{coup}$ (Algorithm~\ref{alg:gcoupling}) which has a competitive ratio of $\tilde{O}(\sqrt{n})$ in the general regime. The new algorithm still combines features from both $\A_{count}$ and $\A_{max}$, but also better estimates the importance of each coordinate. To estimate how much information each coordinate has, the new algorithm divides the samples into $\Theta(\log n)$ groups and checks how often samples from coordinate $i$ are consistent with themselves. If one coordinate has a large proportion of the total information, it uses samples from that coordinate to decide $B$, otherwise it takes a majority vote on samples from all coordinates. 

In Section \ref{sec:topk}, we return to $\topk$ and present an algorithm that has a competitive ratio of $\tilde{O}(\sqrt{n})$, thus proving Theorem~\ref{thm:ubintro}. Our algorithm works by reducing the \topk\ problem to several instances of the \domination\ problem (see Theorem~\ref{thm:main}). At a high level, the algorithm tries to find the top $k$ rows by pairwise comparisons of rows, each of which can be thought of as an instance of $\domination$. We use algorithm $\A_{coup}$ to solve these \domination\ instances. Since we only need to make at most $n^2$ comparisons, if $\A_{coup}$ outputs the correct answer with at least $1-\frac{\eps}{n^2}$ probability for each comparison, then by union bound all the comparisons will be correct with probability at least $1-\eps$. However, to find the top $k$ rows, we do not actually need to compare all the rows to each other; Lemma~\ref{lem:topktournament} shows that we can find the top $k$ rows with high probability while making only $O(n)$ comparisons. Using this lemma, we get a linear time algorithm for solving \topk. Finally in Lemma~\ref{lem:topk_lb}, we extend the lower bound for \domination\ proved in Lemma~\ref{lem:dom_lb} to show a lower bound on the number of samples any algorithm would need on a specific instance of $\topk$. Combining these results, we prove Theorem~\ref{thm:ubintro}.

Finally, in Section \ref{sec:lowerbound}, we show that the algorithms for both $\domination$ and $\topk$ presented in the previous sections have the optimal competitive ratio (up to polylogarithmic factors). Specifically, we show that for any algorithm $A$ solving $\domination$, there exists an instance $C$ of domination where $r_{min}(C, A) \geq \tilde{\Omega}(\sqrt{n})r_{min}(C)$ (Theorem \ref{thm:domlb}). We accomplish this by constructing a distribution $\mathcal{C}$ over instances of $\domination$ such that each instance in the support of this distribution can by solved by an algorithm with low sample complexity (Theorem~\ref{thm:hardub}) but any algorithm that succeeds over the entire distribution requires $\tilde{\Omega}(\sqrt{n})$ times more samples (Theorem~\ref{thm:harddist}). We then embed $\domination$ in $\topk$ (similarly as in Section \ref{sec:topk}) to show an analogous $\tilde{\Omega}(\sqrt{n})$ lower bound for $\topk$ (Theorem ~\ref{thm:sstlb}).

%-----------------------------------------------------------------------------------------------------------------------------

\section{Lower bounds on the sample complexity of domination}
\label{sec:domlb}
%!TEX root =  main.tex

We start by establishing lower bounds on the number of samples $r_{min}(C)$ needed by any algorithm to succeed with constant probability on a given instance $C = (n, \bp, \bq)$ of $\domination$. This is controlled by the quantity $\I(\bp,\bq)$, which is the amount of information we can learn about the bit $B$ given one sample of pairwise comparison between each of the coordinates of $\bp$ and $\bq$. 

\begin{definition}
Given $0\le p,q \le 1$, define $$\I(p,q)=\left(p(1-q)+q(1-p)\right)\left(1-H\left(\frac{p(1-q)}{p(1-q)+q(1-p)}\right)\right).$$
Given $\bp=(p_1,\cdots,p_n)\in [0,1]^n,\bq=(q_1,\cdots,q_n)\in [0,1]^n$, define $$\I(\bp,\bq)=\sum_{i=1}^n \I(p_i,q_i).$$
\end{definition}

%\jonnote{restate in terms of $r_{min}$?}
\begin{lemma}
\label{lem:dom_lb}
Let $C = (n, \bp, \bq)$ be an instance of $\domination$. Then $r_{min}(C) \geq 0.05/\I(\bp, \bq)$. %If $r \leq 0.05/\I(\bp,\bq)$, then no algorithm can successfully solve \domination$(n,\bp,\bq,r)$ with probability at least $2/3$ even if $\bp$ and $\bq$ are known to the algorithm.
\end{lemma}

\begin{proof}
The main idea is to bound the mutual information between the samples and the correct output, and then apply Fano's inequality.
Let $\bp=(p_1,\cdots,p_n)$ and $\bq=(q_1,\cdots,q_n)$.
Recall that $B$ indicates the correct output and that $X_{1,1}, X_{1,2}, ... ,X_{n,r}, Y_{1,1}, ...,Y_{n,r}$ are the samples given to the algorithm.  By Fact \ref{fact:cr},
\[
I(B;X_{1,1}, X_{1,2}, ... ,X_{n,r}, Y_{1,1}, ...,Y_{n,r}) = I(B;X_{1,1}Y_{1,1}) + I(B;X_{1,2}, ... ,X_{n,r}, Y_{1,2}, ...,Y_{n,r}|X_{1,1}Y_{1,1}).
\]
When $\bp$, $\bq$ and $B$ are given, each sample ($X_{i,j}$ or $Y_{i,j}$) is independent of the other samples, and thus $I(X_{1,1} Y_{1,1}; X_{1,2}, ... ,X_{n,r}, Y_{1,2}, ...,Y_{n,r} |B )= 0$. By Fact \ref{fact:it1}, we then have
\[
I(B;X_{1,2}, ... ,X_{n,r}, Y_{1,2}, ...,Y_{n,r}|X_{1,1}Y_{1,1}) \leq I(B;X_{1,2}, ... ,X_{n,r}, Y_{1,2}, ...,Y_{n,r})
\]
and therefore
\[
I(B;X_{1,1}, X_{1,2}, ... ,X_{n,r}, Y_{1,1}, ...,Y_{n,r}) \leq I(B;X_{1,1}Y_{1,1}) + I(B;X_{1,2}, ... ,X_{n,r}, Y_{1,2}, ...,Y_{n,r}).
\]
Repeating this, we get
\[
I(B;X_{1,1}, X_{1,2}, ... ,X_{n,r}, Y_{1,1}, ...,Y_{n,r}) \leq \sum_{i=1}^n\sum_{j =1}^r I(B;X_{i,j}Y_{i,j}).
\]
By Fact \ref{fact:div}, we have
\begin{eqnarray*}
&&I(B;X_{i,j}Y_{i,j}) \\
&=& \Pr[B = 0] \cdot D(X_{i,j}Y_{i,j} | B = 0\| X_{i,j}Y_{i,j}) + \Pr[B = 1] \cdot D(X_{i,j}Y_{i,j} | B = 1 \| X_{i,j}Y_{i,j}) \\
&=& \left(p_i(1-q_i)+q_i(1-p_i)\right)\left(1-H\left(\frac{p_i(1-q_i)}{p_i(1-q_i)+q_i(1-p_i)}\right)\right)\\ 
&=& \I(p_i,q_i). 
\end{eqnarray*}
It follows that
\[
I(B;X_{1,1}, X_{1,2}, ... ,X_{n,r}, Y_{1,1}, ...,Y_{n,r}) \leq \sum_{i=1}^n\sum_{j =1}^r I(B;X_{i,j}Y_{i,j}) = r \cdot \sum_{i=1}^n \I(p_i,q_i) = r\I(\bp,\bq).
\]
For any algorithm, let $p_e$ be its error probability on \domination$(n,\bp,\bq,r)$. By Fano's inequality, we have that
\begin{eqnarray*}
H(p_e) &\geq& H(B|X_{1,1}, X_{1,2}, ... ,X_{n,r}, Y_{1,1}, ...,Y_{n,r})\\
&=& H(B) - I(B;X_{1,1}, X_{1,2}, ... ,X_{n,r}, Y_{1,1}, ...,Y_{n,r})\\
&=&1  - r\I(\bp,\bq) \geq  0.95.
\end{eqnarray*}
Since $H(p_e) \geq 0.95$, we find that $p_e \geq 1/4$, as desired. 
\end{proof}

In the following section, we will concern ourselves with instances $C = (n, \bp, \bq)$ that satisfy $\delta \leq p_i, q_i \leq 1-\delta$ for some constant $\delta$ for all $i$. For such instances, we can approximate $\I(p, q)$ by the $\ell_2$ distance between $\bp$ and $\bq$. 

\begin{lemma}
\label{lem:Ipq_approx_square}
For some $0<\delta\le \frac{1}{2}$, let $\delta\le p,q\le 1-\delta$. Then  $$\frac{1}{4\ln 2} (p-q)^2\le \I(p,q)\le \frac{1}{\delta\ln 2} (p-q)^2.$$
\end{lemma}
\begin{proof}
Let $x=p(1-q)$ and $y=q(1-p)$. Then $\I(p,q)=(x+y)(1-H(\frac{x}{x+y}))$ and $p-q=x-y$. We need to show that
$$(x+y)\left(1-H\left(\frac{x}{x+y}\right)\right)\le \frac{1}{\delta\ln 2} (x-y)^2.$$
By Fact~\ref{fact:repinsker}, $$\frac{1}{\ln 2}z^2\le 1-H\left(\frac{1}{2}+z\right)=D\left(\frac{1}{2}+z \biggr|\biggr| \frac{1}{2}\right)\le \frac{4}{\ln 2}z^2,$$ 
\noindent
and therefore
$$\frac{1}{4\ln 2} \frac{(x-y)^2}{(x+y)}\le (x+y)\left(1-H\left(\frac{x}{x+y}\right)\right)\le  \frac{1}{\ln 2} \frac{(x-y)^2}{(x+y)}.$$
\noindent
Since
$$x+y=p(1-q)+q(1-p)\ge 2\sqrt{p(1-p)q(1-q)}\ge 2\delta (1-\delta)\ge \delta,$$ 
\noindent
this implies the desired upper bound. The lower bound also holds since,
$$x+y=p(1-q)+q(1-p)\le \sqrt{p^2+(1-p)^2}\cdot \sqrt{q^2+(1-q)^2} \le \delta^2+(1-\delta)^2\le 1.$$
\end{proof}

%\jonnote{restate in terms of $r_{min}$?}
\begin{corollary}
\label{cor:dom_lb_half}
Let $C = (n, \bp, \bq)$ be an instance of $\domination$ satisfying $\delta\le \bp_i,\bq_i\le 1-\delta$ for all $i \in [n]$. Then $$r_{min}(C) \geq 0.05 \ln(2) \cdot \frac{\delta}{\norm{\bp-\bq}_2^2}.$$

%Consider \domination$(n,\bp,\bq,r)$ where for every $i\in [n]$, $\delta\le \bp_i,\bq_i\le 1-\delta$ for some $0 < \delta\le \frac{1}{2}$. If $$r \leq 0.05\cdot \ln(2) \cdot \frac{\delta}{\norm{\bp-\bq}_2^2},$$ then no algorithm can succeed to solve \domination$(n,\bp,\bq,r)$ with probability $2/3$ even if $\bp$ and $\bq$ are given to the algorithm.
\end{corollary}
\begin{proof}
By Lemma~\ref{lem:Ipq_approx_square}, $\I(\bp,\bq)\le \norm{\bp-\bq}_2^2/(\delta \ln 2)$. The result then follows from Lemma~\ref{lem:dom_lb}.
\end{proof}

\section{Domination in the well-behaved regime}
\label{sec:d2}
%!TEX root =  main.tex

We now proceed to the problem of designing algorithms for $\domination$ which are competitive on all instances. As a warmup, we begin by considering only instances $C = (n, \bp, \bq)$ of $\domination$ satisfying $\delta \leq p_i,q_i \leq 1- \delta$ for all $i\in [n]$ where $0<\delta\le 1/2$ is some fixed constant. This regime of instances captures much of the interesting behavior of $\domination$, but with the added benefit that the mutual information between the samples and $B$ behaves nicely in this regime: in particular $\I(\bp,\bq) = \Theta( \norm{\bp-\bq}_2^2)$ (see Lemma~\ref{lem:Ipq_approx_square}). By Corollary~\ref{cor:dom_lb_half}, we have $r_{min}\ge \Omega(\frac{1}{\norm{\bp-\bq}_2^2})$. This fact will make it easier to design algorithms for $\domination$ which are competitive in this regime. 

 In Section \ref{d2countmax}, we give two simple algorithms (counting algorithm and max algorithm) which can solve $\domination$ given $\tilde{O}(n/\norm{\bp-\bq}_2^2)$ samples which gives them a competitive ratio of $\tilde{O}(n)$. We will then show that this is tight, i.e. their competitive ratio is $\tilde{\Theta}(n)$ in Lemma~\ref{lem:countingfails} and Lemma~\ref{lem:maxfails}. While the sample complexities of these two algorithms are not optimal, they have the nice property that whenever one performs badly, the other performs well. In Section \ref{d2ub}, we show how to combine the counting algorithm and the max algorithm to give two different algorithms which can solve $\domination$ using only $\tilde{O}(\sqrt{n}/\norm{\bp-\bq}_2^2)$ samples i.e. they have a competitive ratio of $\tilde{O}(\sqrt{n})$. According to Theorem \ref{thm:domlb}, this is the best we can do up to polylogarithmic factors.

\subsection{Counting algorithm and max algorithm}
\label{d2countmax}
%!TEX root =  main.tex

We now consider two simple algorithms for \domination$(n,\bp,\bq)$, which we call the \emph{counting algorithm} (Algorithm \ref{alg:counting}) and the \emph{max algorithm} (Algorithm \ref{alg:max}) denoted by $\A_{count}$ and $\A_{max}$ respectively. We show that both algorithms require $\tilde{O}(\frac{n}{\norm{\bp-\bq}_2^2})$ samples to solve $\domination$ (Lemmas \ref{lem:counting} and \ref{lem:max}). By Corollary~\ref{cor:dom_lb_half}, we have $r_{min}\ge \Omega(\frac{1}{\norm{\bp-\bq}_2^2})$, leading to a $\tilde{O}(n)$ competitive ratio for these algorithms. We show in Lemma \ref{lem:countingfails} and Lemma \ref{lem:maxfails} that this is tight up to polylogarithmic factors i.e. their competitive ratio is $\tilde{\Theta}(n)$.
\begin{algorithm}[ht]
        \caption{The counting algorithm $\mathcal{A}_{count}$ for \domination$(n,\bp,\bq,r)$}
    \begin{algorithmic}[1]\label{alg:counting}
    	\FOR{$i=1$ to $n$}
        \STATE $S_i=\sum_{j=1}^r (X_{i,j} - Y_{i,j})$
        \ENDFOR
        \STATE $Z=\sum_{i=1}^n S_i$
        \STATE If $Z> 0$, output $B=0$. If $Z<0$, output $B=1$. If $Z =0$, output $B=0$ with probability $1/2$ and output $B=1$ with probability $1/2$. 
     \end{algorithmic}
\end{algorithm}

\begin{algorithm}[ht]
        \caption{The max algorithm $\mathcal{A}_{max}$ for \domination$(n,\bp,\bq,r)$}
    \begin{algorithmic}[1]\label{alg:max}
    	\FOR{$i=1$ to $n$}
        \STATE $S_i=\sum_{j=1}^r (X_{i,j} - Y_{i,j})$
        \ENDFOR
        \STATE $i' = \arg\max |S_i|$
        \STATE $Z = S_{i'}$ 
        \STATE If $Z> 0$, output $B=0$. If $Z<0$, output $B=1$. If $Z =0$, output $B=0$ with probability $1/2$ and output $B=1$ with probability $1/2$. 
     \end{algorithmic}
\end{algorithm}

Both the counting algorithm and the max algorithm begin by computing (for each coordinate $i$) the differences between the number of ones in the $X_{i,j}$ samples and $Y_{i, j}$ samples; i.e., we compute the values $S_{i} = \sum_{j=1}^{r} (X_{i, j} - Y_{i, j})$. The counting algorithm $\mathcal{A}_{count}$ decides whether to output $B=0$ or $B=1$ based on the sign of $\sum_{i} S_{i}$, whereas the max algorithm decides its output based on the sign of the $S_i$ with the largest absolute value. See Algorithms \ref{alg:counting} and \ref{alg:max} for detailed pseudocode for both $\mathcal{A}_{count}$ and $\mathcal{A}_{max}$.

We begin by proving upper bounds for the sample complexities of both $\mathcal{A}_{count}$ and $\mathcal{A}_{max}$. In particular, both $\mathcal{A}_{count}$ and $\mathcal{A}_{max}$ need at most $\tilde{O}(n)$ times as many samples as the best possible algorithm for any instance in this regime.

\begin{lemma}
\label{lem:counting}
Let $C = (n, \bp, \bq)$ be an instance of $\domination$. Then 

$$r_{min}(C, \mathcal{A}_{count}, 1-\alpha) \leq \frac{2n\ln (\alpha^{-1})}{\norm{\bp-\bq}_1^2}.$$

\noindent
If $C$ further satisfies $\delta \leq p_i, q_i \leq 1-\delta$ for all $i$ for some constant $\delta>0$, then

$$r_{min}(C, \mathcal{A}_{comb}) \leq O(n)r_{min}(C).$$
\end{lemma}

\begin{proof}
Let $p_e$ be the probability that $B=0$ and $\mathcal{A}_{count}$ outputs $B=1$ when provided with $r =  \frac{2n\ln (\alpha^{-1})}{\norm{\bp-\bq}_1^2}$ samples. By symmetry $p_e$ is equal to the probability that we are in the case $B=1$ and $\mathcal{A}_{count}$ outputs $B=0$ when provided with $r$ samples. It therefore suffices to show that $p_e$ is at most $\alpha$. When $B=0$, 
\[
\E[Z] = \E\left[\sum_{i=1}^n S_i\right] = r\sum_{i=1}^n  (p_i-q_i). 
\]
By the Chernoff bound,
\[
p_e \leq \Pr[Z \leq 0] \leq \exp\left(-\frac{nr}{2} \cdot \left(\frac{\sum_{i=1}^n  (p_i-q_i)}{n}\right)^2\right) \leq \alpha.
\]

The second part of the lemma follows from Corollary \ref{cor:dom_lb_half}, along with the observation that $\norm{\bp-\bq}_1^2 \geq \norm{\bp-\bq}_2^2$.
\end{proof}

\begin{lemma}
\label{lem:max}
Let $C = (n, \bp, \bq)$ be an instance of $\domination$. Then 

$$r_{min}(C, \mathcal{A}_{max}, 1-\alpha) \leq \frac{8\ln (2n\alpha^{-1})}{\norm{\bp-\bq}_\infty^2}$$

\noindent
If $C$ further satisfies $\delta \leq p_i, q_i \leq 1-\delta$ for all $i$ for some constant $\delta$, then

$$r_{min}(C, \mathcal{A}_{comb}) \leq O(n\log n)r_{min}(C).$$
%If $r\ge 8C\log n/\norm{\bp-\bq}_\infty^2$ for some constant $C>1$, then Algorithm~\ref{alg:max} solves \domination$(n,\bp,\bq,r)$ with probability at least $1-2/n^{C-1}$. (Notice $8C\log n/\norm{\bp-\bq}_\infty^2 \leq 8C n\log n/\norm{\bp-\bq}_2^2$)
\end{lemma}
\begin{proof}
Assume without loss of generality that $B=0$, and let $\eps=p_1-q_1=\norm{\bp-\bq}_\infty$. Let $E$ be the event that $\mathcal{A}_{max}$ makes an error and outputs $B=1$ when given $r=\frac{8\ln 2n\alpha^{-1}}{\eps^2}$ samples. We can upper bound the probability of error as $$\Pr[E]\le \Pr[E|S_1> r\eps/2]+\Pr[S_1\le r \eps/2].$$ We will bound each term separately.
Since $\E[S_1]=r(p_1-q_1)=r\eps$, by Hoeffding's inequality,
$$\Pr[S_1\le r\eps/2]\le \exp(-r\eps^2/8)\le \frac{\alpha}{2}.$$
Similarly, by Hoeffding's inequality and the union bound,
$$\Pr[E|S_1>r\eps/2] \le \Pr[\exists i : S_i<-r\eps/2 ]\le n\exp(-r\eps^2/8)\le \frac{\alpha}{2}.$$
It follows that $\Pr[E] \leq \alpha$. The second part of the lemma follows from Corollary \ref{cor:dom_lb_half}, along with the observation that $\norm{\bp - \bq}_2^2 \leq n \norm{\bp - \bq}_{\infty}^2$.
\end{proof}

We now show that the upper bounds we proved above are essentially tight. In particular, we demonstrate instances where both $\mathcal{A}_{count}$ and $\mathcal{A}_{max}$ need $\tilde{\Omega}(n)$ times as many samples as the best possible algorithms for those instances. Interestingly, on the instance where $\mathcal{A}_{count}$ suffers, $\mathcal{A}_{max}$ performs well, and vice versa. This fact will prove useful in the next section.

\begin{lemma}
\label{lem:countingfails}
For each $\eps < \frac{1}{10}$ and each sufficiently large $n$, there exists an instance $C=(n, \bp,\bq)$ of $\domination$ such that the following two statements are true:
\begin{enumerate}
\item $r_{min}(C, \mathcal{A}_{max}, 1 - \frac{2}{n}) \leq \frac{16 \ln n}{\eps^2}$.
\item $r_{min}(C, \mathcal{A}_{count}) \geq \frac{n}{128\eps^2}$.
\end{enumerate}
%\begin{enumerate}
%\item Algorithm \ref{alg:max} solves \domination$(n,\bp,\bq,4\log n/\varepsilon^2)$ with probability at least $1 - 2/n$.
%\item Algorithm \ref{alg:counting} solves \domination$(n,\bp,\bq,r)$ with probability at most $2/3$ when $r \leq n/128\varepsilon^2$. 
%\end{enumerate}
\end{lemma}

\begin{proof}

Let $k$ be an arbitrary integer between 1 and $n-1$. Let $\bp,\bq$ be any vectors satisfying the following constraints:
\begin{enumerate}
\item For all $i \in [n]$, $ \frac{1}{4} < p_i, q_i <  \frac{3}{4}$. 
\item If $i \not\in \{k, k+1\}$, $p_i = q_i$ . 
\item If $i \in \{k, k+1\}$, $q_i = p_i - \eps$ . 
\end{enumerate}
Note that $\norm{\bp-\bq}_\infty =\varepsilon$. Therefore, by Lemma \ref{lem:max}, $r_{min}(C, \mathcal{A}_{max}, 1 - \frac{2}{n}) \leq \frac{16 \ln n}{\eps^2}$, thus proving the first part of the lemma. 

Now assume that $r \leq n/128\varepsilon^2$. We will show that with this many samples, $\mathcal{A}_{count}$ solves instance $C$ with probability at most $3/4$, thus implying the second part of the lemma. Without loss of generality, assume that $B=0$. Define the following random variables $U_{i,j}$:
\begin{enumerate}
\item $U_{i,j} = X_{i,j} - Y_{i,j}$ for $i=1,...,k-1,k+2,...,n$ and $j= 1,..,r$. 
\item $U_{i,j} = X_{i,j} - Y_{i,j} -\varepsilon$. $i = k,k+1$ and $j=1,...,r$. 
\end{enumerate} 
It is straightforward to check that for all $i = 1,...,n, j= 1,..,r$, $\E[U_{i,j}] = 0$, $\E[U^2_{i,j}] \geq 1/4$ and $\E[|U_{i,j}|^3] \leq 1$. Let $\Phi$ be the cdf of the standard normal distribution.
\begin{align*}
&\Pr[\text{$\mathcal{A}_{count}$ outputs $B = 1$ (incorrectly)}]\\
 &= \Pr[\sum_{i=1}^n \sum_{j=1}^r (X_{i,j}-Y_{i,j}) < 0] 
 = \Pr[\sum_{i=1}^n \sum_{j=1}^r U_{i,j} <-2r\varepsilon] \\
 &\geq \Phi\left(-2r\varepsilon\cdot \frac{1}{\sqrt{\sum_{i=1}^n \sum_{j=1}^r \E[U^2_{i,j}]}}\right) - \frac{\sum_{i=1}^n \sum_{j=1}^r \E[|U_{i,j}|^3]}{(\sum_{i=1}^n \sum_{j=1}^r \E[U^2_{i,j}] )^{-3/2}} \tag{By Berry-Esseen theorem (Lemma~\ref{lem:berryesseen})}\\
 &\geq \Phi\left(-\sqrt{\frac{8r\varepsilon^2}{n}}\right) - \frac{8}{\sqrt{nr}}\geq \Phi(-1/4)-\frac{8}{\sqrt{nr}} \geq 1/4.
\end{align*}
\end{proof}

It is not hard to observe that in certain cases, the counting algorithm of \cite{Shah15Sim} for $\topk$ reduces to the algorithm $\mathcal{A}_{count}$ for $\domination$. It follows that there also exists an $\Omega(n)$ multiplicative gap between the sample complexity of their counting algorithm and the sample complexity of the best algorithm on some instances.

\begin{corollary}
\label{cor:countingfails}
Let $A'$ be the $\topk$ algorithm of \cite{Shah15Sim}, and let $S = (n, k, \bP)$ be a $\topk$ instance, with $\bP$ as described in Section \ref{sec:result}. Then, for sufficiently large $n$ and $\eps < 1/10$, $r_{min}(S, A') \geq \Omega(\frac{n}{\eps^2})$.
\end{corollary}

\begin{proof}
Let $i=\pi^{-1}(k)$ and $j=\pi^{-1}(k+1)$. The algorithm $A'$ correctly places $i$ in the set of the top $k$ rows exactly when $\mathcal{A}_{count}$ correctly outputs that row $i$ dominates row $j$. On the other hand, any two consecutive rows of $\bP$ satisfy the constraints in the proof of Lemma \ref{lem:countingfails}. It follows that $r_{min}(S, A') \geq \Omega(\frac{n}{\eps^2})$.
\end{proof}

 We will now show that $\A_{max}$ has a competitive ratio of $\tilde{\Omega}(n)$.
\begin{lemma}
\label{lem:maxfails}
For each sufficiently large $n$, there exists an instance $C = (n, \bp,\bq)$ of $\domination$ such that the following two statements are true:

\begin{enumerate}
\item $r_{min}(C, \mathcal{A}_{count}, 1 - \frac{1}{n}) \leq 2n^3\ln n$.
\item $r_{min}(C, \mathcal{A}_{max}, \frac{4}{5}) \geq \frac{n^4}{2^{14}\ln n}$. 
\end{enumerate}

%\begin{enumerate}
%\item Algorithm \ref{alg:counting} solves \domination$(n,\bp,\bq,2n^3\log n)$ with probability at least $1 - 2/n$.
%\item Algorithm \ref{alg:max} solves \domination$(n,\bp,\bq,r)$ with probability at most $4/5$ when $r \leq \frac{n^4}{2^{14}\ln n}$. 
%\end{enumerate}
\end{lemma}

\begin{proof}
Consider the instance $C = (n, \bp,\bq)$ where $p_i=\frac{1}{2}+\eps$ and $q_i=\frac{1}{2}$, with $\eps=\frac{1}{n^2}$. Since $\norm{\bp-\bq}_1 = \frac{1}{n}$, by Lemma \ref{lem:counting}, $r_{min}(C, \mathcal{A}_{max}, 1 - \frac{1}{n}) \leq 2n^3\ln n$.

%By Corollary~\ref{cor:dom_lb_half}, if $r<\frac{1000}{\norm{\bp-\bq}_2^2}=\frac{n^3}{1000}$, then no algorithm can succeed with probability more than $3/4$. 
Now assume $ r = \frac{n^4}{2^{14}\ln n}$. We will now show that $\mathcal{A}_{max}$ solves $\domination(n,\bp,\bq,r)$ with probability at most $4/5$. Without loss of generality, assume that $B=0$. Define random variables $S_i=\sum_{j=1}^r (X_{i,j}-Y_{i,j})$. Note that $S_1,\cdots,S_n$ are i.i.d random variables with $\E[S_i]=r\eps$ and $\var[S_i]=r(\frac{1}{2}-\eps^2)$. Our algorithm $\mathcal{A}_{max}$ outputs $B=1$ whenever $\inf_i S_i + \sup_i S_i < 0.$ Let $\lambda>0$ be a parameter whose value we will choose later. Note that:

\begin{align*}
\Pr[\inf_i S_i + \sup_i S_i < 0] &\ge \Pr[\inf_i S_i < -\lambda,\ \sup_i S_i < \lambda]\\
&\ge \Pr[\sup_i S_i < \lambda]-\Pr[\inf_i S_i \ge -\lambda,\ \sup_i S_i < \lambda]\\
&=\prod_{i=1}^n\Pr[S_i<\lambda]^n-\prod_{i=1}^n\Pr[-\lambda\le S_i<\lambda]^n\\
&= \Pr[S_1<\lambda]^n-\Pr[-\lambda\le S_1<\lambda]^n\\
&=\Pr[S_1< \lambda]^n-(\Pr[S_1< \lambda]-\Pr[S_1< -\lambda])^n
\end{align*}
We will now apply the Berry-Esseen Theorem (Lemma~\ref{lem:berryesseen}) with $Z_j=(X_{1,j}-Y_{1,j})$ to approximate the CDF of $S_1$. We have $\mu=\E[S_1]=r\eps$, $\sigma^2=\var[S_1]=r(\frac{1}{2}-\eps^2) \ge \frac{r}{4}$. and $\gamma=\sum_{j=1}^r \E[|Z_j-\eps|^3]\le 8r$. Therefore for all $t\in \mathbb{R}$,
$$\left|\Pr[S_1< t]-\Phi\left(\frac{t-\mu}{\sigma}\right)\right|\le \frac{\gamma}{\sigma^3}\le \frac{64}{\sqrt{r}} = \frac{2^{15}\sqrt{\ln n}}{n^{2}}\le \frac{1}{n^{3/2}}$$ when $n$ is large enough. Let us choose $\lambda=\mu+\sigma \Phi^{-1}(1-\frac{\ln 2}{n})$ and let $a=\frac{\lambda-\mu}{\sigma}$, $b=\frac{\lambda+\mu}{\sigma}$. Therefore $\Phi(a)=1-\frac{\ln 2}{n}.$ When $n$ is large enough, $a>10$. By Fact ~\ref{fact:gaussiancdf}, $$\frac{1}{\sqrt{2\pi}}\exp(-a^2/2)\frac{1}{a}\ge \frac{\ln 2}{n}=1-\Phi(a)\ge \frac{1}{\sqrt{2\pi}} \exp(-a^2/2)\frac{1}{2a}.$$ From the left hand side of the above inequality, we can conclude that $a\le 2\sqrt{\ln n}$.
Also,
\begin{align*}
\Phi(-b)=1-\Phi(b) &=\frac{\ln 2}{n} - (\Phi(b)-\Phi(a))\\
&=\frac{\ln 2}{n}-\frac{1}{\sqrt{2\pi}}\int_a^b \exp\left(-t^2/2\right)dt\\
&\ge \frac{\ln 2}{n}-\frac{1}{\sqrt{2\pi}}(b-a)\exp(-a^2/2)\\
&\ge \frac{\ln 2}{n}-\frac{2a(\ln 2)(b-a)}{n}\tag{Since $\frac{1}{\sqrt{2\pi}}\exp(-a^2/2)\frac{1}{2a}\le \frac{\ln 2}{n}$}\\
&\ge \frac{\ln 2}{n}-\frac{4a\mu}{n\sigma}\\
&\ge \frac{\ln 2}{n}-\frac{16\eps \sqrt{r\ln n}}{n} \tag{$\mu=r\eps, \sigma^2\ge \frac{r}{4}, a\le 2\sqrt{\ln n}$}\\
&\ge \frac{\ln 2}{n}-16 \frac{1}{n^2}\frac{1}{n}\sqrt{\frac{n^4}{2^{14}\ln n}\ln n}\\
&\ge \frac{\ln 2}{n}-\frac{1}{8n}
\end{align*}
Now we can bound the probability of error as follows:
\begin{align*}
\Pr[\inf_i S_i + \sup_i S_i < 0] &\ge\Pr[S_1< \lambda]^n-(\Pr[S_1< \lambda]-\Pr[S_1< -\lambda])^n\\
&\ge\left(\Phi\left(\frac{\lambda-\mu}{\sigma}\right)-\frac{1}{n^{3/2}}\right)^n-\left(\Phi\left(\frac{\lambda-\mu}{\sigma}\right)-\Phi\left(\frac{-\lambda-\mu}{\sigma}\right)+2\cdot \frac{1}{n^{3/2}}\right)^n\\
&=\left(\Phi(a)-\frac{1}{n^{3/2}}\right)^n-\left(\Phi(a)-\Phi(-b)+\frac{2}{n^{3/2}}\right)^n\\
&\ge \left(1-\frac{\ln 2}{n}-\frac{1}{n^{3/2}}\right)^n - \left(1-\frac{2\ln 2}{n}+\frac{1}{8n}+\frac{2}{n^{3/2}}\right)^n\\
&\ge \exp(-\ln 2)-\exp(-2\ln 2+1/8)-0.01 \tag{when $n$ is large enough}\\
&> \frac{1}{5}.
\end{align*}
\end{proof}

\subsection{$\tilde{O}(\sqrt{n})$-competitive algorithms}
\label{d2ub}
%!TEX root =  main.tex

We will now demonstrate two algorithms for $\domination$ that use at most $\tilde{O}(\sqrt{n})$ times more samples than the best possible algorithm for each instance. According to Theorem \ref{thm:domlb}, this is the best we can do up to polylogarithmic factors.

Note that the counting algorithm $\mathcal{A}_{count}$ tends to work well when the max algorithm $\mathcal{A}_{max}$ fails, and vice versa (e.g., Lemmas \ref{lem:countingfails} and \ref{lem:maxfails}). Therefore, intuitively, combining both algorithms in some way should lead to better performance. 

Both of the algorithms we present in this section share this intuition. We begin (in Lemma \ref{lem:combining}) by demonstrating a very general method for combining any two algorithms for $\domination$. Applying this to $\mathcal{A}_{count}$ and $\mathcal{A}_{max}$, we obtain an algorithm $\mathcal{A}_{comb}$ that satisfies $r_{min}(C, \mathcal{A}_{comb}) \leq O(\sqrt{n\log n})\cdot r_{min}(C)$ (Corollary \ref{cor:combinecountmax}) for instances $C$ in this regime. We then show an alternate algorithm with slightly better performance than $\mathcal{A}_{comb}$, which we call the \textit{sum of cubes} algorithm $\mathcal{A}_{cube}$. This algorithm satisfies $r_{min}(C, \mathcal{A}_{cube}) \leq O(\sqrt{n})\cdot r_{min}(C)$ for instances $C$ in this regime (Theorem \ref{thm:sumofcubes}). 
\subsubsection{Combining counting and max}

We first show how to combine any two algorithms for $\domination$ to get an algorithm that always does at least as well as the better of the two algorithms. Call an algorithm $\mathcal{A}$ for $\domination$ \emph{stable} if it always outputs the correct answer with probability at least $1/2$ (i.e. it always does at least as well as a random guess). Note that $\mathcal{A}_{count}$ and $\mathcal{A}_{max}$ are both stable. We have the following lemma.

\begin{lemma} \label{lem:combining}
Let $\mathcal{A}_1$ and $\mathcal{A}_2$ be two stable algorithms for $\domination$. Then there exists an algorithm $\mathcal{A}_{comb}$ such that for all instances $C$ of $\domination$,

$$r_{min}(C, \mathcal{A}_{comb}, 1-\alpha) \leq 32\ln (\alpha^{-1}) \cdot \min\left(r_{min}(C, \mathcal{A}_1), r_{min}(C, \mathcal{A}_2)\right)$$
\end{lemma}

\begin{proof}
See Algorithm \ref{alg:combining} for a description of $\mathcal{A}_{comb}$. Assume without loss of generality that $B=0$, and let $r=32\log(n\alpha^{-1})\min\left(r_{min}(C, \mathcal{A}_1), r_{min}(C, \mathcal{A}_2)\right)$. We will show that $\mathcal{A}_{comb}$ outputs $B=0$ correctly with probability at least $1-\alpha$. 

Let $r' = \frac{r}{32\ln n}$; note that either $r' \geq r_{min}(C, \mathcal{A}_1)$ or $r_{min}(C, \mathcal{A}_2)$. Assume first that $r' \geq r_{min}(C, \mathcal{A}_1)$. Then, $\mathcal{A}_{1}$ will output $B=0$ in each of its $16\ln \alpha^{-1}$ groups with probability at least $\frac{3}{4}$. On the other hand, since it is stable, $\mathcal{A}_{2}$ will output $B=0$ in each of its groups with probability at least $\frac{1}{2}$. Therefore

$$\E\left[\frac{Z_1+Z_2}{2}\right]\le  \frac{1}{8}+\frac{1}{4}\le \frac{3}{8}.$$

Since $\frac{Z_1+Z_2}{2}$ is the average of $32\ln \alpha^{-1}$ random variables, by Hoeffding's inequality, the probability that $\frac{Z_1+Z_2}{2} \geq \frac{1}{2}$ is at most $\exp\left(-2(32\ln \alpha^{-1})(\frac{1}{8})^2\right) \leq \alpha$. 

Similarly, if $r' \geq r_{min}(C, \mathcal{A}_2)$, the probability that $\frac{Z_1+Z_2}{2} \geq \frac{1}{2}$ is also at most $\alpha$. This concludes the proof.
\end{proof}

\begin{algorithm}[ht]
        \caption{Combining two algorithms $\mathcal{A}_1$ and $\mathcal{A}_2$ for \domination$(n,\bp,\bq,r)$}
    \begin{algorithmic}[1]\label{alg:combining}
        \STATE Divide the samples into $32\ln \alpha^{-1}$ groups.
        \STATE Run $\mathcal{A}_1$ on each of the first $16\ln \alpha^{-1}$ groups and let $Z_1$ be the average of the outputs.
        \STATE Run $\mathcal{A}_2$ on each of the last $16\ln \alpha^{-1}$ groups and let $Z_2$ be the average of the outputs.
        \STATE If $\frac{Z_1+Z_2}{2}\le \frac{1}{2}$ output $B=0$, else output $B=1$.
     \end{algorithmic}
\end{algorithm}

\begin{corollary}
\label{cor:combinecountmax}
Let $\mathcal{A}_{comb}$ be the algorithm we obtain by combining $\mathcal{A}_{count}$ and $\mathcal{A}_{max}$ in the manner of Lemma \ref{lem:combining}. Then for any instance $C = (n, \bp, \bq)$ of $\domination$,

$$r_{min}(C, \mathcal{A}_{comb}) \leq O\left(\frac{\sqrt{n\log n}}{\norm{\bp-\bq}_2^2}\right).$$

\noindent
If $C$ further satisfies $\delta \leq p_i, q_i \leq 1-\delta$ for all $i$ for some constant $\delta$, then

$$r_{min}(C, \mathcal{A}_{comb}) \leq O(\sqrt{n\log n})r_{min}(C).$$
\end{corollary}
\begin{proof}
This follows from Lemmas \ref{lem:counting}, \ref{lem:max}, \ref{lem:combining}, and the following observation:

\begin{eqnarray*}
\min\left(\frac{n}{\norm{\bp-\bq}_1^2},\frac{\log n}{\norm{\bp-\bq}_\infty^2}\right) &\le & \sqrt{\frac{n}{\norm{\bp-\bq}_1^2}\cdot \frac{\log n}{\norm{\bp-\bq}_\infty^2}} \\
&\le & \frac{\sqrt{n\log n}}{\norm{\bp-\bq}_2^2}.
\end{eqnarray*}

The last inequality follows from the fact that for any vector $\mathbf{x}$, $\norm{\mathbf{x}}_2^2 \leq \norm{\mathbf{x}}_1 \cdot \norm{\mathbf{x}}_{\infty}$. The second part of the corollary then follows directly from Corollary \ref{cor:dom_lb_half}.
\end{proof}

\subsubsection{The sum of cubes algorithm}

We now give a different algorithm for $\domination$ which we call the \emph{sum of cubes algorithm}, $\mathcal{A}_{cube}$. If we let $S_i = \sum_{j} (X_i - Y_i)$, then intuitively, whereas $\mathcal{A}_{count}$ decides its output based on the signed $\ell_1$ norm of the $S_i$ and whereas $\mathcal{A}_{max}$ decides its output based on the signed $\ell_{\infty}$ norm of the $S_i$, $\mathcal{A}_{cube}$ decides its output based on the signed $\ell_{3}$ norm of the $S_i$. See Algorithm~\ref{alg:sumofcubes} for a detailed description of the algorithm.

\begin{algorithm}[ht]
        \caption{Sum of cubes algorithm $\mathcal{A}_{cube}$ for \domination$(n,\bp,\bq,r)$}
    \begin{algorithmic}[1]\label{alg:sumofcubes}
        \STATE $T_{i,j}=1$ with probability $\frac{1}{2}+\frac{(X_{i,j}-Y_{i,j})}{2}$ and $T_{i,j}=-1$ with probability $\frac{1}{2}-\frac{(X_{i,j}-Y_{i,j})}{2}$
        \STATE $S_i=\sum_{j=1}^r T_{i,j}$
        \STATE $Z=\sum_{i=1}^n S_i^3$
        \STATE If $Z\ge 0$, output $B=0$. If $Z<0$, output $B=1$.
     \end{algorithmic}
\end{algorithm}

To analyze the performance of $\mathcal{A}_{cube}$, we begin by analyzing statistical properties of the random variable $S$.
%We write $A\gtrsim B$ when there is some absolute constant $c>0$ such that $A\ge c\cdot B$. We define $A \lesssim B$ in a similar way. We write $A\gg B$ to denote that $A\ge C\cdot B$ for some sufficiently large constant $C$.
\begin{lemma}
\label{lem:thirdmoments}
Let $S=\sum_{j=1}^r X_{j}$ where $X_{1},\cdots,X_{r}$ are i.i.d $\{-1,1\}$-valued random variables with mean $\eps\ge 0$ and $r\ge 8$. Let $Z=S^3$. Then
\begin{align*}
\E[Z]& \ge 2r^2 \eps +\frac{1}{2} r^3 \eps^3 \\
\var[Z]& \le 15 r^3+36 r^4\eps^2+9r^5\eps^4.
\end{align*}
\end{lemma}

\begin{proof}
By applying the multinomial theorem and using the fact that $X_i^2=1$ for each $i$, we can write multilinear expressions for $S^3$ and $S^6$. We can now use linearity of expectation and the independence among the $X_i$'s to compute the mean and variance exactly.
\begin{align*}
\E[Z]&=\E[S^3]=(-2 r + 3 r^2) \eps + (2 r - 3 r^2 + r^3) \eps^3 \ge 2r^2 \eps +\frac{1}{2} r^3 \eps^3 \\
\var[Z]&=E[S^6]-\E[S^3]^2=(16 r - 30 r^2 + 15 r^3) + (-136 r + 282 r^2 - 183 r^3 + 36 r^4) \eps^2 +\\
&\hspace{1cm} (240r - 522 r^2 + 381 r^3 - 108 r^4 + 9 r^5) \eps^4 + (-120 r + 270 r^2 - 213 r^3 + 72r^4 - 9 r^5) \eps^6\\
& \le 15r^3+36r^4\eps^2+9r^5\eps^4
\end{align*}
\end{proof}

\begin{lemma}
\label{lem:sumofcubes}
Let $S_i=\sum_{j=1}^r X_{i,j}$ where for each $i\in [n]$, $X_{i,1},\cdots,X_{i,r}$ are i.i.d $\{-1,1\}$-valued random variables with mean $\eps_i$, along with the condition that either all $\eps_i \ge 0$ or all $\eps_i \le 0$. Let $Z=\sum_{i=1}^n S_i^3$. If $r\ge 8$ and $r\ge \eta \sqrt{n}/(\sum_{i=1}^n \eps_i^2)$ for some $\eta \ge 1$ then, $\E[Z]^2\ge \frac{\eta}{36} \var[Z]$.
\end{lemma}
\begin{proof}
Without loss of generality, we can assume that $\eps_i\ge 0$ for all $i\in [n]$.
By Lemma~\ref{lem:thirdmoments}, 
\begin{align}
\E[Z]^2&\ge 4 r^4(\sum_i \eps_i )^2+ \frac{1}{4} r^6 (\sum_i \eps_i^3)^2 + 2r^5 (\sum \eps_i)(\sum_i \eps_i^3)\label{eqn:expsquare}\\
\var[Z]&\le 15nr^3 +36 r^4\sum_i \eps_i^2 +9 r^5 \sum_i \eps_i^4. \label{eqn:variance}
\end{align}

We will show that each term in the Equation~\ref{eqn:variance} is dominated by some term in Equation~\ref{eqn:expsquare}. 
\begin{align*}
nr^3 &= r^5 \frac{n}{r^2} \le \frac{1}{\eta^2}r^5 (\sum_i \eps_i ^2)^2 \le \frac{1}{\eta^2}r^5 (\sum_i \eps_i)(\sum_i \eps_i^3) \tag{Cauchy-Schwarz inequality}\\
r^4(\sum_i \eps_i^2) &\le \frac{1}{\eta\sqrt{n}} r^5 (\sum_i\eps_i^2)^2 \le \frac{1}{\eta\sqrt{n}} r^5 (\sum_i \eps_i)(\sum_i \eps_i^3)\\
r^5(\sum_i \eps_i^4) &\le r^6 \frac{1}{\eta\sqrt{n}} (\sum_i \eps_i^2)(\sum_i \eps_i^4) \le r^6 \frac{1}{\eta\sqrt{n}} \left(\sqrt{n}\cdot (\sum_i \eps_i^4)^{1/2}\right)(\sum_i \eps_i^4) \tag{Cauchy-Schwarz inequality}\\
& \hspace{4cm}= \frac{r^6}{\eta} (\sum_i \eps_i^4)^{3/2} \le \frac{r^6}{\eta} (\sum_i \eps_i^3)^2 \tag{monotonicity of $\ell_p$ norms}
\end{align*}

Adding the above inequalities, we get $\var[Z]\le \frac{36}{\eta} \E[Z]^2$.
\end{proof}

\begin{theorem}
\label{thm:sumofcubes}
If $C = (n, \bp, \bq)$ is any instance of $\domination$, then

$$r_{min}(C, \mathcal{A}_{cube}) \leq \max\left(\frac{144\sqrt{n}}{\norm{\bp-\bq}_2^2},8\right).$$

\noindent
If $C$ satisfies $\delta \leq p_i, q_i \leq 1-\delta$ for all $i$ for some constant $\delta$, then

$$r_{min}(C, \mathcal{A}_{cube}) \leq O(\sqrt{n})r_{min}(C).$$
\end{theorem}
\begin{proof}
%\[
%\E[T_{i,j}]=\E[X_{i,j}-Y_{i,j}]=\left\{
%\begin{array}{lr}
%p_i-q_i,& B=0\\
%q_i-p_i,& B=1
%\end{array}.\right.
%\]
Assume without loss of generality that $B=0$. We have $S_i=\sum_{j=1}^r T_{i,j}$ and $Z=\sum_{i=1}^n S_i^3$. Note that for each $i$, the $T_{i,j}$ are i.i.d. $\{-1, 1\}$ random variables with mean $\E[T_{i,j}] = p_i - q_i$. Applying Lemma~\ref{lem:sumofcubes} with $\eta=144$, if $r \geq \max\left(\frac{144\sqrt{n}}{\norm{\bp-\bq}_2^2},8\right)$ we have that $\E[Z]^2\ge \frac{108}{36} \var[Z]=\frac{\var[Z]}{3}$.	
Since the algorithm makes an error (i.e. outputs $B=1$) when $Z< 0$, we can use Chebyshev's inequality to bound the probability that $Z<0$.
$$\Pr[Z< 0]\le \Pr[|Z-\E[Z]|\ge \E[Z]]\le \frac{\var[Z]}{\E[Z]^2}\le \frac{1}{4}.$$

The second part of the theorem then follows directly from Corollary \ref{cor:dom_lb_half}.
\end{proof}

%-----------------------------------------------------------------------------------------------------------------------------
\section{Domination in the general regime}\label{sec:domination_general}
In this section, we consider \domination~in the general regime. Unlike in the previous section, it is no longer true that $I(X_{i,j}Y_{i,j} ; B) =\I(p_i,q_i)= \Theta( (p_i-q_i)^2)$. In particular, when $p_i$ and $q_i$ are both very small, $\I(p_i,q_i)$ can be much bigger than $(p_i-q_i)^2$; as a result, the algorithms designed in the previous section can fail under these circumstances. 

In Section \ref{dub}, we present a new algorithm which is $\tilde{O}(\sqrt{n} \cdot r_{min})$-competitive. According to Theorem \ref{thm:domlb}, this is the best we can do up to polylogarithmic factors. In Section \ref{infgap}, we then demonstrate that the general regime is indeed harder than the restricted regime in Section \ref{sec:d2}. In particular, we give instances where the algorithms presented in the previous section fail; we show that the competitive ratio of these algorithms is unbounded (even for fixed $n$). 

\subsection{An $\tilde{O}(\sqrt{n})$-competitive algorithm}
\label{dub}
%!TEX root =  main.tex
Here we give an algorithm that only needs $O(\sqrt{n}\log(n)/\I(\bp,\bq))$ samples to solve $\domination$ (Theorem \ref{thm:main}). By Lemma \ref{lem:dom_lb}, this is only $\tilde{O}(\sqrt{n})$ times as many samples as the optimal algorithm needs. Intuitively, the algorithm works as follows: if for some coordinate $i$, $X_{i,1}Y_{i,1}...X_{i,r},Y_{i,r}$ conveys enough information about $B$, we will only use samples from coordinate $i$ to determine $B$. Otherwise, the information about $B$ must be well-spread throughout all the coordinates, and a majority vote will work.

We begin by bounding the probability we can determine the answer from a single fixed coordinate.

\begin{lemma}[Sanov's theorem]
\label{lem:sanovthm}
Let $\cP(\Sigma)$ denote the space of all probability distributions on some finite set $\Sigma$. Let $R\in \cP(\Sigma)$ and let $Z_1,\cdots,Z_k$ be i.i.d random variables with distribution $R$. For every $x\in \Sigma^k$, we can define an empirical probability distribution $\hP_x$ on $\Sigma$ as $$\forall \sigma\in\Sigma\quad \hP_x(\sigma)=\frac{|\{i\in [k]: x_i=\sigma\}|}{k}.$$ Let $C$ be a closed convex subset of $\cP(\Sigma)$ such that for some $P\in C$, $D(P||R)<\infty$. Then 
$$\Pr\left[\hP_{(Z_1,\cdots,Z_k)}\in C\right]\le \exp\left(-k(\ln 2) D(Q^*||R)\right)$$ where $Q^*=\argmin_{Q\in C} D(Q||R)$ is unique. In the case when $D(Q||R)=\infty$ for all $Q\in C$,  $\Pr\left[\hP_{(Z_1,\cdots,Z_k)}\in C\right]=0.$
\end{lemma}
\begin{proof}
See exercise 2.7 and 3.20 in \cite{CK11}.
\end{proof}

\begin{lemma} 
\label{lem:sanovoptimaltest}
Let $0\le q<p\le 1$ and let $X_1,\cdots,X_k$ be i.i.d $\Ber(p)$ and $Y_1,\cdots,Y_k$ be i.i.d $\Ber(q)$. Then $$\Pr\left[\sum_{i=1}^k (X_i-Y_i) \le 0\right] \le \exp\left(-2(\ln 2)k\log\left(\frac{1}{\sqrt{pq}+\sqrt{(1-p)(1-q)}}\right)\right).$$
\end{lemma}
\begin{proof}
We will use Sanov's theorem (Lemma~\ref{lem:sanovthm}). Let $\Sigma=\{0,1\}^2$. Consider the set of distributions on $\Sigma$, $$\cP(\Sigma)=\{(p_{00},p_{01},p_{10},p_{11}): 0\le p_{00},p_{01},p_{10},p_{11}\le 1, p_{00}+p_{01}+p_{10}+p_{11}=1\},$$
and define $C\subset \cP(\Sigma)$ as $C=\{(p_{00},p_{01},p_{10},p_{11}): p_{01}\ge p_{10}\}$. Clearly $C$ is a closed convex set. Define $R=\left((1-p)(1-q),(1-p)q,p(1-q),pq\right) \in \cP(\Sigma)$; note that this is exactly the distribution of $(X_i,Y_i)$ for each $i\in [k]$. Since $p>q$, $R\notin C$. Observe that $\sum_{i=1}^r (X_i-Y_i)\le 0$ iff the empirical distribution generated by $(X_1,Y_1),\cdots,(X_k,Y_k)$, $\hP_{\left((X_1,Y_1),\cdots,(X_k,Y_k)\right)}$ belongs to  $C$. We can assume that there is some $Q\in C$ such that $D(Q||R)<\infty$, otherwise the lemma is trivially true. Therefore by Lemma~\ref{lem:sanovthm}, $$\Pr\left[\sum_{i=1}^r (X_i-Y_i)\le 0\right]\le \exp\left(-k(\ln 2)D(Q^*||R)\right)$$ where $Q^*=\argmin_{Q\in C} D(Q||R)$ is unique. In addition, $Q^*$ should lie on the boundary of $C$ i.e. $Q^*$ should satisfy $p_{01}=p_{10}$. So $$D(Q^*||R)=\min_{0\le x,y\le 1,\ x+2y\le 1} D((1-x-2y,y,y,x)||R).$$ Let $f(x,y)=(\ln 2)D((1-x-2y,y,y,x||R)$. Since $D(Q||R)$ is convex as a function of $Q$, $f(x,y)$ is convex as well. We will show that there is always a point in the region $\{0\le x,y\le 1,\ x+2y\le 1\}$ where the gradient of $f(x,y)$ is zero. Since $f$ is convex, this must be the minimizer of $f$. Note that
\begin{align*}
\frac{\partial f(x,y)}{\partial x}&= -1 -\ln(1-x-2y)+\ln((1-p)(1-q))+1+\ln x-\ln(pq)=0\\
\frac{\partial f(x,y)}{\partial y}&= -2 -2\ln(1-x-2y)+\ln((1-p)(1-q))+2+2\ln y-\ln(pq)=0.
\end{align*}
Solving the above equations for $x,y$ we get $$x=\frac{pq}{\left(\sqrt{pq}+\sqrt{(1-p)(1-q)}\right)^2},\quad y=\frac{\sqrt{pq(1-p)(1-q)}}{\left(\sqrt{pq}+\sqrt{(1-p)(1-q)}\right)^2}.$$
It is easy to check that $0\le x,y\le 1$ and $x+2y\le 1$. Substituting the values of $x,y$, we find that $$D(Q^*||R)=-2\log\left(\sqrt{pq}+\sqrt{(1-p)(1-q)}\right).$$
\end{proof}

\begin{lemma}
\label{lem:ublbdiv_approx}
$$2\log\left(\frac{1}{\sqrt{pq}+\sqrt{(1-p)(1-q)}}\right)\ge \frac{1}{2} \I(p,q).$$
\end{lemma}
\begin{proof}
We can assume $0<p,q<1$, otherwise the required inequality follows from the fact that $-\ln(1-t)\ge t$ for $0\le t< 1$. For example, when $p=0$, the LHS simplifies to $-\log(1-q)$ and the RHS to $q/2$, and the inequality is satisfied. The other cases are similar. Hence, from now on, assume that $0<p,q<1$.
Let $x=p(1-q)$ and $y=q(1-p)$. Thus $\I(p,q)=(x+y)(1-H(x/x+y))$. We can also the write the LHS of the inequality as:
\begin{align*}
-2\log\left(\sqrt{pq}+\sqrt{(1-p)(1-q)}\right)&=-\log\left(pq+(1-p)(1-q)+2\sqrt{pq(1-p)(1-q)}\right)\\
&=-\log\left(1-x-y+2\sqrt{xy}\right)\\
&=-\log\left(1-(\sqrt{x}-\sqrt{y})^2\right)\\
&\ge  (\sqrt{x}-\sqrt{y})^2/(\ln 2) \tag{$-\log(1-t)\ge t/(\ln 2)$}\\
&= (x+y-2\sqrt{xy})/(\ln 2).
\end{align*}
Now we need to show that $$\frac{(x+y-2\sqrt{xy})}{\ln 2}\ge \frac{1}{2}(x+y)\left(1-H\left(\frac{x}{x+y}\right)\right).$$
We can scale $x,y$ such that $x+y=1$, so let $x=\frac{1}{2}+z$ and $y=\frac{1}{2}-z$. Therefore it is enough to show that
 $$1-\sqrt{1-4z^2}\ge \frac{\ln 2}{2} \left(1-H\left(\frac{1}{2}+z\right)\right).$$
We have $1-\sqrt{1-4z^2}\ge 2z^2$ and by Fact~\ref{fact:repinsker}, $$1-H\left(\frac{1}{2}+z\right)=D\left(\frac{1}{2}+z \biggr|\biggr| \frac{1}{2}\right)\le \frac{4}{\ln 2}z^2.$$ Combining these two, we have the required inequality.
 
\end{proof}

\begin{lemma}
\label{lem:sanovsalgorithm)}
In \domination$(n,\bp,\bq,r)$, for any $i\in [n]$, if $r > 6/\I(p_i,q_i)$, then $$\Pr\left[\sgn\left(\sum_{j=1}^r (X_{i,j} -Y_{i,j})\right)=(-1)^B\right] > 5/6.$$ 
\end{lemma}
\begin{proof}
Assume we are in the $B=0$ case, the other case is similar. Fix an $i\in [n]$. By Lemma~\ref{lem:sanovoptimaltest}, 
\begin{align*}
\Pr\left[\sum_{j=1}^r (X_{i,j} -Y_{i,j})\le 0\right] &\le \exp\left(-r(\ln 2)\log\left(\frac{1}{\sqrt{p_iq_i}+\sqrt{(1-p_i)(1-q_i)}}\right)\right)\\
&\le \exp\left(-r(\ln 2)I_i/2\right) \tag{By Lemma~\ref{lem:ublbdiv_approx}}\\
&=2^{-rI_i/2}< 1/8.
\end{align*}
\end{proof}

We now introduce what we call the \textit{general coupling algorithm} $\mathcal{A}_{coup}$ for $\domination$. A detailed description of the algorithm can be found in Algorithm \ref{alg:gcoupling}; more briefly the algorithm works as follows:

\begin{enumerate}
\item Split the $r$ samples for each of the $n$ coordinates into $\ell = 18\log(2n\alpha^{-1})$ equally-sized segments where $\alpha$ is the error parameter. For each coordinate $i$ and segment $j$, set $S_{i, j} = 1$ if more samples from $X$ equal $1$ than samples from $Y$, and $-1$ otherwise. This can be thought of as running a miniature version of the counting algorithm on each segment; $S_{i, j} = 1$ is evidence that $B=0$, and $S_{i, j} = -1$ is evidence that $B=-1$.

\item Let $i'$ be the coordinate $i$ which maximizes $\left|\sum_{j=1}^{\ell}S_{i,j}\right|$ (i.e. the coordinate that is ``most consistently'' either $1$ or $-1$). If $\left|\sum_{j=1}^{\ell}S_{i', j}\right| \geq \ell/3$ (i.e. at least $2\ell/3$ of the segments for this coordinate agree on the value of $B$), output $B$ according to the sign of $\sum_{j=1}^{\ell}S_{i', j}$.

\item Otherwise, for each segment, take the majority of the votes from each of the $n$ coordinates; that is, for each $1\leq j\leq \ell$, set $T_{j} = \sgn(\sum_{i=1}^{n}S_{i,j})$. Then take another majority over the segments, by setting $Z_2 = \sgn(\sum_{j=1}^{\ell}T_j)$. Finally, if $Z_2 > 0$ output $B=0$; otherwise, output $B=1$.
\end{enumerate}

\begin{algorithm}[ht]
        \caption{General coupling algorithm for \domination$(n,\bp,\bq,r)$}
    \begin{algorithmic}[1]\label{alg:gcoupling}
        \STATE $\ell = 18\log(2n\alpha^{-1})$. 
    	\FOR{$i=1$ to $n$}
	\FOR{$j = 1$ to $\ell$}
        \STATE $S_{i,j}=\sgn (\sum_{t=(j-1)*(r/\ell) +1}^{jr/\ell} X_{i,t} - Y_{i,t})$
        \STATE If $S_{i,j} = 0$, let $S_{i,j} = 1$ with probability $1/2$ and let $S_{i,j} = -1$ with probability $1/2$. 
        \ENDFOR
        \ENDFOR
        \STATE $i' = \arg\max_i |\sum_{j=1}^{\ell} S_{i,j}|$
        \STATE $Z_1 = \sum_{j=1}^{\ell} S_{i',j}$
        \IF {$|Z_1| \geq \ell/3$}
        		\STATE If $Z_1>0$ output $B=0$, else output $B=1$. 
	\ELSE
        \FOR{$j=1$ to $l$}
        \STATE $T_j = \sgn(\sum_{i=1}^n S_{i,j})$. 
        \STATE If $T_{j} = 0$, let $T_{j} = 1$ with probability $1/2$ and let $T_{j} = -1$ with probability $1/2$. 
        \ENDFOR
        \STATE $Z_2 = \sgn(\sum_{j=1}^{\ell} T_j)$. 
        
        \STATE If $Z_2 = 0$, let $Z_2 = 1$ with probability $1/2$ and let $Z_2 = -1$ with probability $1/2$. 
        		\STATE If $Z_2>0$ output $B=0$, else output $B=1$. 
        \ENDIF
     \end{algorithmic}
\end{algorithm}

\begin{theorem}
\label{thm:main}
If $C = (n, \bp, \bq)$ is any instance of $\domination$, then

$$r_{min}(C, \mathcal{A}_{coup}, 1-\alpha) \leq \frac{2592\sqrt{n}\ln (2n\alpha^{-1})}{\I(\bp, \bq)}$$

\noindent
and thus

$$r_{min}(C, \mathcal{A}_{coup}) \leq O(\sqrt{n}\log n)\cdot r_{min}(C).$$
%There exists absolute constant $C > 0$, such that for \domination$(n,\bp,\bq, r)$, if $r \geq C\sqrt{n}\log(n)/\I(\bp,\bq)$, then Algorithm \ref{alg:gcoupling} outputs $B$ correctly with probability at least $1- \frac{2}{n}$.
\end{theorem}

\begin{proof}
%Let $C = 36\cdot 24 \cdot 6 = 5184$ \jonnote{$C$ is the instance, change to something else}.
Let $I_i=\I(p_i,q_i)$, $r=2592\sqrt{n}\log (2n\alpha^{-1})/\I(\bp, \bq)$ and  $\ell = 18\ln (2n\alpha^{-1})$. There are two cases to consider:
\begin{enumerate}
\item \textbf{Case 1}: There exists an $i'$ such that $24 \sqrt{n} I_{i'} \geq \sum_{k=1}^n I_k$. 

By symmetry, we can assume that $B =0$. In this case, we have that $\frac{r}{\ell} \geq \frac{24 \cdot 6\sqrt{n}}{\sum_{k=1}^n{I_k}} \geq \frac{6}{I_{i'}}$. By Lemma \ref{lem:sanovsalgorithm)}, for each $j=1,\dots,\ell$, $\Pr[S_{i',j} =1] \geq 5/6$. Therefore we have
\[
\E\left[ \sum_{j=1}^l S_{i',j}\right] \geq \ell \cdot (5/6 - 1/6) = 2\ell/3.
\]
Since $S_{i',1}, ...,S_{i',l}$ are independent when $B$ is given, by the Chernoff bound, we have that
\[
\Pr\left[ \sum_{j=1}^l S_{i',j} \geq \ell/3\right] \geq 1 - \exp( -\ell \cdot (1/3)^2 \cdot (1/2)) \geq 1 - \frac{\alpha}{2n}.
\]
For $i \neq i'$, since $p_i \geq q_i$, we still have $\Pr[S_{i,j} =1] \geq 1/2$. By a similar argument, we get
\[
\Pr\left[ \sum_{j=1}^l S_{i,j} \geq -\ell/3\right] \geq 1 - \exp(-\ell \cdot (1/3)^2 \cdot (1/2)) \geq 1 - \frac{\alpha}{2n}.
\]
Let $W$ be the event that $\sum_{j=1}^{\ell} S_{i',j} \geq \ell/3$ and for $i\neq i'$, $\sum_{j=1}^l S_{i,j} \geq -\ell/3$. By the union bound, we have that $\Pr[W] \geq 1 - n\cdot \frac{\alpha}{2n} = 1-\frac{\alpha}{2}$. Moreover, when $W$ happens, we know that $Z_1 \geq \ell/3$ and $\mathcal{A}_{coup}$ outputs $B=0$.  Therefore, in Case 1, the probability that $\mathcal{A}_{coup}$ outputs $B$ correctly is at least $1-\frac{\alpha}{2}$. 

\item \textbf{Case 2}: For all $i \in \{1,\dots,n\}$, $24 \sqrt{n} I_{i} < \sum_{k=1}^n I_k$. 

Similarly as in Case 1, since $\Pr[S_{i,j} =(-1)^B] \geq 1/2$, the probability that $|Z_1| \geq \ell/3$ and our algorithm outputs wrongly is at most $\frac{\alpha}{2}$. For the rest of Case 2, assume $|Z_1| < \ell/3$. 

Now fix a coordinate $i$. Our plan is to first lower bound the amount of information samples from coordinate $i$ have about $B$ by using Lemma \ref{lem:sanovsalgorithm)} and the subadditivity of information. Let $s = r/\ell$, and let $s' = s \cdot \lceil \frac{6}{s I_i}\rceil$. Imagine that we have $s'$ new samples, $U_{i,1},V_{i,1},..., U_{i,s'},V_{i,s'}$, where each $(U_{i,j},V_{i,j})$ ($j=1,\dots,s'$) is generated independently according to the same distribution as $(X_{i,1},Y_{i,1})$. Since $s' \geq 6/I_i$, by Lemma  \ref{lem:sanovsalgorithm)}, we have that
\[
\Pr\left[\sgn\left(\sum_{j=1}^{s'} (U_{i,j} -V_{i,j})\right)=(-1)^B\right] > 5/6.
\]

Write $(U_{i}V_{i})^{[a, b]}$ as shorthand for the sequence $((U_{i,a}, V_{i, a}), \dots (U_{i,b}, V_{i,b}))$, and define  $(X_{i}Y_{i})^{[a,b]}$ analogously. By Fano's inequality, we have that
\begin{eqnarray*}
I\left((U_{i}V_{i})^{[1,s']}; B\right) &=& H(B) - H(B|(U_{i}V_{i})^{[1,s']})\\
 &\geq& H(\tfrac{1}{2}) - H(1 - \tfrac{5}{6}) = 1 - H(\tfrac{1}{6}) \geq 1/3.
\end{eqnarray*}
Since $I((U_{i}V_{i})^{[1, s]}; (U_{i}V_{i})^{[s+1, s']} | B) = 0$ (our new samples are independent given $B$),  we have
\begin{eqnarray*}
I((U_{i}V_{i})^{[1,s']}; B) &=& I((U_{i}V_{i})^{[1,s]}; B |(U_{i}V_{i})^{[s+1,s']}) + I((U_iV_i)^{[s+1, s']} ;B)\\
&\leq& I((U_{i}V_{i})^{[1,s]}; B) + I((U_iV_i)^{[s+1, s']} ;B) ~~~~~~~~(\text{by Fact \ref{fact:it1}})
\end{eqnarray*}
Repeating this procedure, we get
\[
 I((U_{i}V_{i})^{[1,s']}; B)  \leq \sum_{u=1}^{\lceil \frac{6}{s I_i}\rceil} I((U_{i}V_{i})^{[(u-1)s+1, us]}; B).
\]
Since we know that for any $u=1,...,\lceil \frac{6}{s I_i}\rceil$, 
\[
I((U_{i}V_{i})^{[(u-1)s+1, us]}; B) = I((X_{i}Y_{i})^{[1, s]}; B),
\]
we get
\[
I((X_{i}Y_{i})^{[1, s]}; B) \geq  I((U_{i}V_{i})^{[1, s']}; B) \cdot \frac{1}{\lceil \frac{6}{s I_i}\rceil} \geq  \frac{s I_i}{6\cdot 6}.
\]
The last inequality is true because $\frac{6}{s I_i} = \frac{ \sum_{k=1}^n {I_k}}{24\sqrt{n}I_i} \geq 1$. 

After we lower bound $I((X_{i}Y_{i})^{[1, s]}; B)$, we are going to show that we can output $B$ correctly with reasonable probability based on samples only from coordinate $i$. 
\begin{eqnarray*}
\frac{s I_i}{6 \cdot 6} &\leq& I((X_{i}Y_{i})^{[1, s]}; B) \\
&=& \sum_x \Pr[(X_{i}Y_{i})^{[1, s]} = x] \cdot D(B|(X_{i}Y_{i})^{[1, s]} = x \|  B) \\
&\leq& \sum_x \Pr[(X_{i}Y_{i})^{[1, s]} = x]  \cdot \big( 2(\Pr[B = 0|(X_{i}Y_{i})^{[1, s]} = x ] -1/2)^2 +\\
&& \;\;\; 2(\Pr[B = 1|(X_{i}Y_{i})^{[1, s]} = x ] -1/2)^2 \big)~~~~~~~~~~~~~~~~~~~(\text{by Fact \ref{fact:repinsker}})\\
&=& \sum_x \Pr[(X_{i}Y_{i})^{[1, s]} = x]  \cdot  (\Pr[B = 0|(X_{i}Y_{i})^{[1, s]} = x ] -\Pr[B = 1|(X_{i}Y_{i})^{[1, s]} = x ])^2 \\
&\leq& \sum_x \Pr[(X_{i}Y_{i})^{[1, s]} = x]  \cdot \left|\Pr[B = 0|(X_{i}Y_{i})^{[1, s]} = x ] -\Pr[B = 1|(X_{i}Y_{i})^{[1, s]} = x ]\right|. \\
\end{eqnarray*}
When $\sum_{j=1}^{s} (X_{i,j} - Y_{i,j}) > 0$, it is easy to check that 
\[
\Pr[B = 0|(X_{i}Y_{i})^{[1, s]}] > \Pr[B = 1|(X_{i}Y_{i})^{[1, s]}].
\]
Therefore,
\begin{eqnarray*}
\Pr[S_{i,1} = (-1)^B] &=& \sum_x \Pr[(X_{i}Y_{i})^{[1, s]} = x] \cdot \\
&&\;\;\; \max\left(\Pr[B = 0|(X_{i}Y_{i})^{[1, s]} = x ], \Pr[B = 1|(X_{i}Y_{i})^{[1, s]} = x ]\right) \\
&=&  \frac{1}{2} + \frac{1}{2} \cdot \sum_x \Pr[(X_iY_i)^{[1,s]} = x] \cdot \\ 
&& \left|\Pr[B = 0|(X_{i}Y_{i})^{[1, s]} = x ] - \Pr[B = 1|(X_{i}Y_{i})^{[1, s]} = x ]\right| \\
&\geq& \frac{1}{2} + \frac{s I_i}{12\cdot 6} \\
&\geq& \frac{1}{2} + \frac{\sqrt{n} I_i}{\sum_{k=1}^n I_k}. \\
\end{eqnarray*}
Similarly, we can show for all $i = 1,...,n$, $j = 1,...,l$, 
\[
\Pr[S_{i,j} = (-1)^B] \geq \frac{1}{2} + \frac{\sqrt{n} I_i}{\sum_{k=1}^n I_k}. 
\]
Now without loss of generality assume that $B = 0$. We have that
\[
\E\left[ \sum_{i=1}^n S_{i,j} \right] \geq \sum_{i=1}^n \left( \frac{1}{2} + \frac{\sqrt{n} I_i}{\sum_{k=1}^n I_k}  -  \frac{1}{2} + \frac{\sqrt{n} I_i}{\sum_{k=1}^n I_k}\right) = 2\sqrt{n} .
\]
Therefore, by the Chernoff bound, 
\[
\Pr[ T_j = 1 ] \geq 1 - e^{- (1/n) \cdot (2\sqrt{n})^2 \cdot (1/2) } > 3 /4. 
\]
By the Chernoff bound again,
\[
\Pr[Z_2 > 0] \geq 1- e^{-\ell\cdot (1/2)^2 \cdot(1/2)} \geq 1 - \frac{\alpha}{2n}. 
\]
Since we initially fail with probability at most $\frac{\alpha}{2}$, by the union bound, in Case 2 we fail with probability at most $\frac{\alpha}{2} + \frac{\alpha}{2n} < \alpha$. This concludes the proof.
\end{enumerate}
\end{proof}

\subsection{$\mathcal A_{count}$ and $\mathcal A_{max}$ have unbounded competitive ratios}
\label{infgap}
%!TEX root =  main.tex

In this section, we show that the competitive ratio of $\mathcal A_{count}$ and $\mathcal A_{max}$ is unbounded. The result in Lemma \ref{lem:countingfails2} can be easily generalized to show that the counting algorithm of \cite{Shah15Sim} for \topk\ also has unbounded competitive ratio.

\begin{lemma}
\label{lem:countingfails2}
For each sufficiently large $n$ and for any $\eps >0$, there exists an instance $C = (n,\bp,\bq)$ of $\domination$ such that the following two statements are true:
\begin{enumerate}
\item $r_{min}(C, \mathcal{A}_{coup}, 1 - \frac{2}{n}) \leq \frac{5184\sqrt{n}\log n}{\eps}$
\item $r_{min}(C, \mathcal{A}_{count}) \geq \frac{n}{16\eps^2}$.
\end{enumerate}
\end{lemma}

\begin{proof}
Define $\bp,\bq$ to be
\begin{enumerate}
\item $p_1 = \eps$, $q_1 = 0$.
\item $p_i = q_i = 1/2$ for $i=2,...,n$. 
\end{enumerate}
%Recall in  Theorem \ref{thm:main}, we define
%\[
%I_i = \left(p_i(1-q_i)+q_i(1-p_i)\right)\left(1-H\left(\frac{p_i(1-q_i)}{p_i(1-q_i)+q_i(1-p_i)}\right)\right).
%\]
Note that $\I(p_1,q_1) = \eps$, and $\I(p_2,q_2) = \cdots = \I(p_n,q_n) = 0$. Therefore, by Theorem \ref{thm:main}, $\mathcal{A}_{coup}$ succeeds given $r =  \frac{5184\sqrt{n}\log n}{\eps}$ samples with probability at least $1 - 2/n$. 

Now assume that $r \leq n/(16 \varepsilon^2) \leq  (n-1)/(8 \varepsilon^2)$. We will now show that $\mathcal{A}_{count}$ solves \domination$(n,\bp,\bq,r)$ with probability at most $3/4$. Without loss of generality assume that $B=0$. Define the random variables $U_{i,j}$ as follows:
\begin{enumerate}
\item $U_{1,j} = X_{1,j} - Y_{1,j} -\varepsilon$. $j=1,...,r$. 
\item $U_{i,j} = X_{i,j} - Y_{i,j}$ for $i = 2,...,n, j= 1,..,r$. 
\end{enumerate} 
It is straightforward to check that for all $i = 1,\dots,n$ and $j= 1,\dots,r$, $\E[U_{i,j}] = 0$ and $\E[|U_{i,j}|^3] \leq 1/2$. For all $i =2,\dots,n$ and $j = 1,\dots,r$, we further have that $\E[U^2_{i,j}] = 1/2$. Let $\Phi$ be the cdf of the standard normal distribution.
\begin{align*}
&\Pr[\text{$\mathcal{A}_{count}$ outputs $B=1$ (incorrectly)}]\\
 &= \Pr[\sum_{i=1}^n \sum_{j=1}^r (X_{i,j}-Y_{i,j}) < 0]
 = \Pr[\sum_{i=1}^n \sum_{j=1}^r U_{i,j} <-r \cdot \varepsilon] \\
 &\geq \Phi\left(-r \cdot \varepsilon \cdot \frac{1}{\sqrt{\sum_{i=1}^n \sum_{j=1}^r \E[U^2_{i,j}]}}\right) - \frac{\sum_{i=1}^n \sum_{j=1}^r \E[|U_{i,j}|^3]}{(\sum_{i=1}^n \sum_{j=1}^r \E[U^2_{i,j}] )^{-3/2}} \tag{By Berry-Esseen theorem (Lemma~\ref{lem:berryesseen})}\\
 &\geq \Phi\left(-r \cdot \varepsilon / \sqrt{r (n-1)/2} \right) - \sqrt{\frac{2n^2}{(n-1)^3 r}}\geq \Phi(-1/4)- \sqrt{\frac{2n^2}{(n-1)^3 r}} \geq 1/4.
\end{align*}
\end{proof}

%--------------------------------------------------------------------------------------------------------------

\begin{lemma}
\label{lem:maxfails2}
For each sufficiently large $n$ and any $0 < \eps < 1/n^3$, there exists an instance $C=(n, \bp,\bq)$ of $\domination$ such that the following two statements are true.
\begin{enumerate}
\item $r_{min}(C, \mathcal{A}_{coup}, 1 - \frac{2}{n}) \leq \frac{518400\sqrt{n}\ln n}{\eps}$. 

\item $r_{min}(C, \mathcal{A}_{max}, \frac{9}{10}) \geq \frac{1}{\eps^{2}2^{14}\ln n}$
\end{enumerate}
\end{lemma}

\begin{proof}
Define $\bp,\bq$ as:
\begin{enumerate}
\item $p_1 = \varepsilon / 100$, $q_1 = 0$.
\item $p_i = 1/2 + \varepsilon$, $q_i = 1/2$ $i=2,...,n$. 
\end{enumerate}
Note that $\I(p_1,q_1) = \varepsilon / 100$ and $\I(p_2,q_2) = \cdots \I(p_n,q_n)= (1-H(1/2+\varepsilon))/2 $. By Fact \ref{fact:repinsker}, $\I(p_2,q_2)\leq \frac{4}{\ln(2)} \cdot (\varepsilon)^2 \leq \varepsilon/(100n)$. Thus $\varepsilon/100 \leq \sum_{i=1}^n \I(p_i,q_i)\leq \varepsilon/50$. Therefore, by Theorem \ref{thm:main}, given at least $\frac{518400\sqrt{n}\ln n}{\eps}$ samples, $\mathcal{A}_{coup}$ succeeds with probability at least $1 - 2/n$. 

Now fix $r = \frac{1}{\varepsilon^2 2^{14}\ln n}$.  We will now show that $\mathcal{A}_{max}$ solves \domination$(n,\bp,\bq,r)$ with probability at most $9/10$. Without loss of generality assume $B=0$. Define random variable $S_i=\sum_{j=1}^r (X_{i,j}-Y_{i,j})$. $S_1$ is always non-negative. $S_2,\cdots,S_n$ are i.i.d random variables with $\E[S_i]=r\eps$ and $\var[S_i]=r(\frac{1}{2}-\eps^2)$. Algorithm \ref{alg:max} outputs $B=1$ when $\inf_i S_i + \sup_i S_i < 0.$ Let $\lambda>0$ be some parameter which we will choose later.

\begin{align*}
\Pr[\inf_i S_i + \sup_i S_i < 0] &\ge \Pr[\inf_i S_i < -\lambda,\ \sup_i S_i < \lambda]\\
&\ge \Pr[\sup_i S_i < \lambda]-\Pr[\inf_i S_i \ge -\lambda,\ \sup_i S_i < \lambda]\\
&=\prod_{i=1}^n\Pr[S_i<\lambda]^n-\prod_{i=1}^n\Pr[-\lambda\le S_i<\lambda]^n\\
&= \Pr[S_1 < \lambda] \left(\Pr[S_2<\lambda]^{n-1}-\Pr[-\lambda\le S_2<\lambda]^{n-1}\right)\\
&=\Pr[S_1 < \lambda]  \left(\Pr[S_2< \lambda]^{n-1}-(\Pr[S_2< \lambda]-\Pr[S_2< -\lambda])^{n-1}\right)\\
\end{align*}

We will now apply Berry-Esseen Theorem (Lemma~\ref{lem:berryesseen}) with $Z_j=(X_{2,j}-Y_{2,j})$ for $j=1,\cdots,r$, to approximate the CDF of $S_2$. We have $\mu=\E[S_2]=r\eps$, $\sigma^2=\var[S_2]=r(\frac{1}{2}-\eps^2) \ge \frac{r}{4}$. and $\gamma=\sum_{j=1}^r \E[|Z_j-\eps|^3]\le 8r$. Therefore for all $t\in \mathbb{R}$,
$$\left|\Pr[S_2< t]-\Phi\left(\frac{t-\mu}{\sigma}\right)\right|\le \frac{\gamma}{\sigma^3}\le \frac{64}{\sqrt{r}} \le \frac{64}{n^{3/2}}$$
when $n$ is large enough.
Let us choose $\lambda=\mu+\sigma \Phi^{-1}(1-\frac{\ln 2}{n-1})$ and let $a=\frac{\lambda-\mu}{\sigma}$, $b=\frac{\lambda+\mu}{\sigma}$. Therefore $\Phi(a)=1-\frac{\ln 2}{n-1}.$ When $n$ is large enough, $a>10$. By Fact ~\ref{fact:gaussiancdf}, $$\frac{1}{\sqrt{2\pi}}\exp(-a^2/2)\frac{1}{a}\ge \frac{\ln 2}{n-1}=1-\Phi(a)\ge \frac{1}{\sqrt{2\pi}} \exp(-a^2/2)\frac{1}{2a}.$$ From the left hand side of the above inequality, we can conclude that $a\le 2\sqrt{\ln (n-1)}$.
Also,
\begin{align*}
\Phi(-b)=1-\Phi(b) &=\frac{\ln 2}{n} - (\Phi(a)-\Phi(b))\\
&=\frac{\ln 2}{n-1}-\frac{1}{\sqrt{2\pi}}\int_a^b \exp\left(-t^2/2\right)dt\\
&\ge \frac{\ln 2}{n-1}-\frac{1}{\sqrt{2\pi}}(b-a)\exp(-a^2/2)\\
&\ge \frac{\ln 2}{n-1}-\frac{2a(\ln 2)(b-a)}{n-1}\\
&\ge \frac{\ln 2}{n-1}-\frac{4a\mu}{(n-1)\sigma}\\
&\ge \frac{\ln 2}{n-1}-\frac{16\eps \sqrt{r\ln (n-1)}}{n-1} \tag{$\mu=r\eps, \sigma^2\ge \frac{r}{4}, a\le 2\sqrt{\ln (n-1)}$}\\
&\ge \frac{\ln 2}{n-1}-16 \frac{\varepsilon}{n-1}\sqrt{\frac{1}{\varepsilon^22^{14}\ln n}\ln (n-1)}\\
&\ge \frac{\ln 2}{n-1}-\frac{1}{8(n-1)}
\end{align*}

By Chernoff bound, we have
\[
\Pr[S_1 < \lambda] \geq \Pr[S_1 \leq \mu] = 1- e^{-r \cdot D(\varepsilon \| \varepsilon /100)} \geq 1 - e^{-\frac{2.5}{\varepsilon}\cdot \varepsilon}\geq 1/2.
\]
Now we can bound the probability of error as follows:
\begin{align*}
&\Pr[\inf_i S_i + \sup_i S_i < 0] \\
&\ge \Pr[S_1 < \lambda]  \left(\Pr[S_2< \lambda]^{n-1}-(\Pr[S_2< \lambda]-\Pr[S_2< -\lambda])^{n-1}\right)\\
&\ge\frac{1}{2}\left( \left(\Phi\left(\frac{\lambda-\mu}{\sigma}\right)-\frac{64}{n^{3/2}}\right)^{n-1}-\left(\Phi\left(\frac{\lambda-\mu}{\sigma}\right)-\Phi\left(\frac{-\lambda-\mu}{\sigma}\right)+2\cdot \frac{64}{n^{3/2}}\right)^{n-1}\right)\\
&=\frac{1}{2}\left(\left(\Phi(a)-\frac{64}{n^{3/2}}\right)^{n-1}-\left(\Phi(a)-\Phi(-b)+\frac{128}{n^{3/2}}\right)^{n-1}\right)\\
&\ge\frac{1}{2}\left( \left(1-\frac{\ln 2}{n-1}-\frac{64}{n^{3/2}}\right)^{n-1} - \left(1-\frac{2\ln 2}{n-1}+\frac{1}{8(n-1)}+\frac{128}{n^{3/2}}\right)^{n-1}\right)\\
&\ge\frac{1}{2}\left( \exp(-\ln 2)-\exp(-2\ln 2+1/8)-0.01\right) \tag{when $n$ is large enough}\\
&> \frac{1}{10}.
\end{align*}
\end{proof}

%-----------------------------------------------------------------------------------------------------------------------------

\section{Reducing top-$K$ to domination}\label{sec:topk}
%!TEX root =  main.tex

In this section, we will finally reduce $\topk$ to $\domination$, thus proving Theorem~\ref{thm:ubintro}. First, we will give an algorithm for $\topk$ problem that uses $\mathcal{A}_{coup}$ for $\domination$ as a subroutine. We begin by reducing $\topk$ to the following graph theoretic problem.
\begin{lemma}
\label{lem:topktournament}
Let $G=([n],E)$ be a directed complete graph on vertices $\{1,2,\cdots,n\}$ i.e. for every distinct $i,j\in [n]$, either $(i,j)\in E$ or $(j,i)\in E$ but not both. Suppose there is a subset $S\subset [n]$ of size $k$ such that $(i,j)\in E$ for every $i\in S$ and $j\notin S$. Then there is a randomized algorithm which runs in expected running time $O(n)$ and finds the set $S$ given oracle access to the edges of $G$. Moreover there is some absolute constant $C>0$ such that for every $\lambda\ge 1$, the probability that the algorithm runs in more than $C\lambda n$ time is bounded by $\exp(-\lambda)$.
\end{lemma}
\begin{proof}
Pick $v\in [n]$, uniformly at random. Let $d^{in}(v)$ and $d^{out}(v)$ be the indegree and outdegree of vertex $v\in [n]$. Clearly $d^{in}(v)+d^{out}(v)=n-1$. Also $v\in S$ iff $d^{in}(v)< k$. We can thus easily test if $v\in S$ by querying the $n-1$ edges, $\{(i,v):i\in [n]\setminus\{v\}\}$. Depending on whether $v \in S$, we now have two cases:

\begin{itemize}
\item \textbf{Case 1}: $v\in S$

For every $i$ such that $(i,v)\in E$, we can conclude that $i\in S$. We can therefore remove these vertices and iterate. We have reduced the problem to a graph on $n-1-d^{in}(v)=d^{out}(v)$ vertices.
\item \textbf{Case 2}: $v\notin S$

For every $i$ such that $(v,i)\in E$, we can conclude that $i\notin S$. We can therefore remove these vertices and iterate. We have reduced the problem to a graph on $n-1-d^{out}(v)=d^{in}(v)$ vertices.
\end{itemize}
Let $n'$ be the number of vertices that remain after the above random process. Note that
\begin{align*}
\E_v[n']&=\Pr[v\in S]\cdot\E[d^{out}(v)|v\in S]+\Pr[v\notin S]\cdot\E[d^{in}(v)|v\notin S]\\
&=\frac{k}{n}\left(n-k+\frac{k-1}{2}\right)+\frac{n-k}{n}\left(k+\frac{n-k-1}{2}\right)\\
&=\frac{n-1}{2}+\frac{k(n-k)}{2n} \le \frac{3n}{4}.
\end{align*}
By Markov's inequality, $\Pr[n'\ge 4n/5]\le \frac{15}{16}$. We will repeatedly choose $v$ at random until we find a $v$ such that $n'< \frac{4n}{5}$. Once we find such a $v$, we can remove at least $n/5$ vertices from the graph and iterate the same procedure for the remaining graph. Let $T_0$ denote the random variable equal to the number of times we sample $v$. We have that $\Pr[T_0\ge t]\le (\frac{15}{16})^t$ and therefore $$\E[T_0]=\sum_{t=1}^\infty \Pr[T_0\ge t] \le 15.$$ Similarly let $T_i$ represent the number of times we must sample $v$ in iteration $i$ of this process; by the same logic, $\E[T_i] \leq 15$ for all $i$. If we let the random variable $X$ denote the number of edge queries the algorithm makes, then since the graph shrinks by a factor of $4/5$ at each iteration, 

\begin{align*}
X&=T_0\cdot n+T_1\cdot \left(\frac{4}{5}\right)n+T_2\cdot \left(\frac{4}{5}\right)^2n+\cdots \\
\E[X] &\leq 15 \cdot \left(1 + \frac{4}{5} + \left(\frac{4}{5}\right)^2 + \dots\right) \cdot n \leq 75 n.
\end{align*}

This completes the proof that  $\E[X]=O(n)$, as required. We can similarly analyze the tail probability of $X$; note that:
$$\Pr[X>C \lambda n]\le \Pr\left[\exists i: T_i > \frac{\lambda C}{9} \left(\frac{10}{9}\right)^i\right]$$
since $T_i\le \frac{C\lambda}{9} \left(\frac{10}{9}\right)^i$ for every $i$ implies that $$X\le \frac{C\lambda n}{9} \sum_{i=0}^\infty \left(\frac{4}{5}\right)^i\left(\frac{10}{9}\right)^i=\frac{C\lambda n}{9}\sum_{i=0}^\infty \left(\frac{8}{9}\right)^i\le C\lambda n.$$ By the union bound,
\begin{align*}
\Pr\left[\exists i: T_i > \frac{C\lambda }{9} \left(\frac{10}{9}\right)^i\right]&\le \sum_{i=0}^\infty \Pr\left[T_i>\frac{C\lambda}{9} \left(\frac{10}{9}\right)^i \right]\\
&\le \sum_{i=0}^\infty \exp\left(-\frac{C\lambda }{9} \ln\left(\frac{16}{15}\right)\left(\frac{10}{9}\right)^i \right)\\
&\le \exp(-\lambda). \tag{for sufficiently large $C$}
\end{align*}
\end{proof}
The following lemma shows that when $p \geq q$, $\I(p,q)$ is an increasing function of $p$ and a decreasing function of $q$.
\begin{lemma}
\label{lem:information_monotone}
Let $0\le q'\le q\le p\le p'\le 1$, then $\I(p',q')\ge \I(p,q)$.
\end{lemma}
\begin{proof}
We have:
\begin{align*}
\frac{\partial \I(p,q)}{\partial p}&=(1-q)\log\left(\frac{2p(1-q)}{p(1-q)+(1-p)q}\right)-q\log\left(\frac{2(1-p)q}{p(1-q)+(1-p)q}\right)\\
\frac{\partial \I(p,q)}{\partial q}&=(1-p)\log\left(\frac{2(1-p)q}{p(1-q)+(1-p)q}\right)-p\log\left(\frac{2p(1-q)}{p(1-q)+(1-p)q}\right)
\end{align*}
When $p\ge q$, $$\log\left(\frac{2p(1-q)}{p(1-q)+(1-p)q}\right)\ge 0,\quad \log\left(\frac{2(1-p)q}{p(1-q)+(1-p)q}\right)\le 0.$$ Thus $\frac{\partial \I(p,q)}{\partial p}\ge 0$ and $\frac{\partial \I(p,q)}{\partial q}\le 0$ when $p\ge q$. Thus increasing $p$ or decreasing $q$ cannot decrease $\I(p,q)$ when $p\ge q$.
\end{proof}
We are now ready to give an algorithm for $\topk$.
\begin{theorem}
\label{thm:algmain}
There exists an algorithm $A$ for $\topk$ such that for any $\alpha > 0$ and any instance $S = (n, k, \bP)$, $A$ runs in time $O(n^2r\log(1/\alpha))$ and satisfies

$$r_{min}(S, A, 1-\alpha) \leq \frac{7776\sqrt{n}\log(2n\alpha^{-1})}{\I(\bP_k, \bP_{k+1})}$$

\noindent
where $\bP_k,\bP_{k+1}$ are the $k$ and $k+1$ rows of $\bP$.
\end{theorem}
\begin{proof}
Let $\bP_i$ denote the $i^{th}$ row of $\bP$, and let $\Delta = I(\bP_{k}, \bP_{k+1})$.  Recall that $\mathcal{A}$ is given as input the three-dimensional array of samples $Z_{i, j, l}$, where for each $i, j \in [n]$ and $1 \leq l \leq r$, $Z_{i, j, l}$ is the result of the $l$th noisy comparison between item $i$ and item $j$ (sampled from $\Ber(\bP_{\pi^{-1}(i), \pi^{-1}(j)})$). We will define a complete directed graph $G=([n],E)$ as follows. For every $1 \leq i < j \leq n$, run $\mathcal{A}_{coup}$ with input $X_{h, l} = Z_{i, h, l}$ and $Y_{h, l} = Z_{j, h, l}$; if $\mathcal{A}_{coup}$ returns $B=0$, then direct the edge from $i$ towards $j$, and otherwise, direct the edge from $j$ towards $i$.

Let $T = \{\pi(1), \pi(2), \dots, \pi(k)\}$ be the set of labels of the top $k$ items. We claim that if $i \in T$ and $j \not\in T$, then with probability at least $1 - \frac{\alpha}{n^2}$, the edge is directed from $i$ towards $j$. To see this, note that in the corresponding input to $\mathcal{A}_{coup}$, $X$ is drawn from $\bP_{\pi^{-1}(i)}$ and $Y$ is drawn from $\bP_{\pi^{-1}(j)}$. If $i \in T$ and $j \not\in T$, then $\pi^{-1}(i) \leq k < \pi^{-1}(j)$. In particular, $\bP_{\pi^{-1}(i)}$ dominates $\bP_{\pi^{-1}(j)}$, and moreover by Lemma~\ref{lem:information_monotone}, $\I(\bP_{\pi^{-1}(i)},\bP_{\pi^{-1}(j)})\ge \Delta$. It follows from Theorem \ref{thm:main} that $A_{coup}$ outputs $B=0$ on this input with probability at least $1 - \frac{\alpha}{2n^2}$, since in general,

\begin{eqnarray*}
r_{min}(C, \mathcal{A}_{coup}, 1-\tfrac{\alpha}{2n^2}) &\leq& \frac{2592\sqrt{n}\log (4n^3\alpha^{-1})}{\I(\bp, \bq)} \\
&\leq& \frac{7776\sqrt{n}\log (2n\alpha^{-1})}{\I(\bp, \bq)}.
\end{eqnarray*}

By the union bound, the probability that all of these comparisons are correct is at least $1-\frac{\alpha}{2}$. Therefore, by the tail bounds in Lemma~\ref{lem:topktournament}, we can find the subset $T$ in $O(n\log(1/\alpha))$ oracle calls to $\mathcal{A}_{coup}$ with probability at least $1-\frac{\alpha}{2}$. The probability of failure is at most $\frac{\alpha}{2} + \frac{\alpha}{2} = \alpha$. Each call to Algorithm~\ref{alg:gcoupling} takes $O(nr)$ time, so the overall time of the algorithm is $O(n^2r\log(1/\alpha))$.

%We are given samples according to some matrix $\bQ$ which is obtained by permuting the rows and colums of $\bP$ using some random permutation $\pi:[n]\to [n]$ i.e. $\bQ_{i,j}=\bP_{\pi^{-1}(i),\pi^{-1}(j)}$. Define a complete directed graph $G=([n],E)$ as follows: The edge between $i,j$ is directed from $i$ to $j$ if $\mathcal{A}_{coup}$ given samples distributed according to rows $\bQ_i,\bQ_j$ predicts that $\bQ_i>\bQ_j$ and vice versa. Note that $\I(\bQ_i,\bQ_j)=\I(\bP_{\pi^{-1}(i)},\bP_{\pi^{-1}(j)})$. Therefore for every $i\le k, j>k$, we have $\I(\bQ_{\pi(i)},\bQ_{\pi(j)})\ge \Delta.$ Therefore by Theorem~\ref{thm:main}, Algorithm~\ref{alg:gcoupling} can solve \domination$(n,\bQ_{\pi(i)},\bQ_{\pi(j)},r)$ correctly with probability $\ge 1-1/n^C$ for every $i\le k, j>k$ i.e. it will predict $\bQ_{\pi(i)}>\bQ_{\pi(j)}$. By union bound, Algorithm~\ref{alg:gcoupling} doesn't make a mistake on any pair of rows $\pi(i),\pi(j)$ where $i\le k, j>k$ with probability at least $1-\frac{1}{n^{C-2}}$. So with probability $1-\frac{1}{n^{C-2}}$, the graph $G$ has a subset $S\in [n]$ of size $k$ such that for every $i\in S$ and $j\notin S$,  $(i,j)\in E$. This subset $S$ is exactly $\{\pi(1),\cdots,\pi(k)\}$. Therefore by Lemma~\ref{lem:topktournament}, we can find the subset $S$ in $O(n\log(1/\eps))$ oracle calls to Algorithm~\ref{alg:gcoupling} with probability $1-\eps/2$. The probability of failure is at most $\eps/2+\frac{1}{n^{C-2}}\le \eps$ when $C\ge 3$. Each call to Algorithm~\ref{alg:gcoupling} takes $O(nr)$ time, so the overall time of the algorithm is $O(n^2r\log(1/\eps))$.
\end{proof}

To prove that this algorithm is competitive, we will conclude by proving a lower bound on $r_{min}(S)$ (again, by reduction to the appropriate lower bound for $\domination$). 
\begin{lemma}
\label{lem:topk_lb}
Let $S = (n, k, \bP)$ be an instance of $\topk$. Then $r_{min}(S) \geq \frac{0.1}{\I(\bP_{k}, \bP_{k+1})}$. 
\end{lemma}

\begin{proof}
We will proceed by contradiction. Suppose there exists an algorithm $A$ which satisfies $r_{min}(S, A) \leq \frac{0.01}{\I(P_{k}, P_{k+1})}$. We will show how to convert this into an algorithm $A'$ which solves the instance $C = (n, \bP_{k}, \bP_{k+1})$ of $\domination$ with probability at least $\frac{3}{4}$ when given at least $2r = 0.05/\I(\bP_{k}, \bP_{k+1})$ samples, thus contradicting Lemma \ref{lem:dom_lb}.

The algorithm $A'$ is described in Algorithm \ref{alg:lbreduction}; essentially, $A'$ embeds the inputs $X$ and $Y$ to the $\domination$ instance as rows/columns $k$ and $k+1$ respectively of the $\topk$ instance. It is easy to check that the $Z_{i,j,l}$ for $i,j \in[n]$,$l \in [r]$ generated in $A'$ are distributed according to the same distribution as the corresponding elements in the instance $S$ of $\topk$. Therefore $A$ will output the top $k$ items correctly with probability at least $3/4$. In addition, if $B =0$ the item labeled $k$ will be in the top $k$ items and if $B=1$ the item labeled $k$ will not be in the top $k$ items. Therefore, $A'$ succeeds to solve this instance of $\domination$ with probability at least $3/4$, leading to our desired contradiction.

\begin{algorithm}[ht]
    \caption{Algorithm $A'$ for the lower bound reduction}
    \begin{algorithmic}[1]\label{alg:lbreduction}
\STATE Get input $X_{i,l}$, $Y_{i,l}$ for $i\in[n]$ and $l \in [2r]$ from \domination$(n,\bP_k,\bP_{k+1},2r)$.
\STATE Generate a random permutation $\pi$ on $n$ elements s.t. $\pi(\{k, k+1\}) = \{k, k+1\}$.
\FOR{ $i \in [n]$, $j \in [n], l\in [r]$}
	\STATE If $i = k$, set $Z_{i,j,l} = X_{j,l}$. 
	\STATE If $i = k+1$, set $Z_{i,j,l} = Y_{j,l}$.
	\STATE If $i \not \in \{k,k+1\}, j = k$, set $Z_{i,j,l} = X_{i,l+r}$.
	\STATE If $i \not \in \{k,k+1\}, j = k+1$, set $Z_{i,j,l} = Y_{i,l+r}$.
	\STATE If $i \not \in \{k,k+1\}, j\not \in \{k , k+1\}$, sample $Z_{i,j,l}$ from $\Ber(\bP_{\pi^{-1}(i), \pi^{-1}(j)})$.
\ENDFOR
\STATE Run $A$ on samples $Z_{i,j,l}$, $i,j \in [n]$, $l \in [r]$. 
\STATE If $A$ said $k$ is amongst the top $k$ items, output $B=0$. Otherwise output $B=1$. 
\end{algorithmic}
\end{algorithm}

\end{proof}

We are now ready to prove our main upper bound result.

\begin{corollary}
\label{cor:ubmain}
There is an algorithm $A$ for \topk\ such that $A$ runs in time  $O(n^2r)$ and on every instance $S$ of \topk\ on $n$ items,
\[
r_{min}(S, A) \leq O(\sqrt{n}\log n)r_{min}(S).
\]
\end{corollary}

\begin{proof}
Let $S=$\topk$(n,k,\bP,\cdot)$ be an instance of \topk. By Lemma~\ref{lem:topk_lb}, $$r_{\min}(S)\ge \frac{0.1}{\I(\bP_k,\bP_{k+1})}.$$
If $A$ is the algorithm in Theorem~\ref{thm:algmain} with $\alpha = \frac{1}{4}$ then $A$ runs in time $O(n^2r)$ and
$$r_{\min}(S,A)\le O\left(\frac{\sqrt{n}\log n}{\I(\bP_k,\bP_{k+1})}\right).$$ 
Combining these two inequalities, we obtain our result.

\end{proof}

%-----------------------------------------------------------------------------------------------------------------------------

\section{Hardness of domination and top-$K$}\label{sec:lowerbound}
%!TEX root =  main.tex

In the previous section we demonstrated an algorithm that solves $\topk$ on any distribution using at most $\tilde{O}(\sqrt{n})$ times more samples than the optimal algorithm for that distribution (see Corollary \ref{cor:ubmain}). In this section, we show this is tight up to logarithmic factors; for any algorithm, there exists some distribution where that algorithm requires $\tilde{\Omega}(\sqrt{n})$ times more samples than the optimal algorithm for that distribution. Specifically, we show the following lower bound.

\begin{theorem}\label{thm:sstlb}
For any algorithm $A$, there exists an instance $S$ of $\topk$ of size $n$ such that $r_{min}(S,A) \geq \Omega\left(\frac{\sqrt{n}}{\log n}\right) r_{min}(S)$.
\end{theorem}

As in the previous sections, instead of proving this lower bound directly, we will first prove a lower bound for the domination problem, which we will then embed in a $\topk$ instance.

\begin{theorem}\label{thm:domlb}
For any algorithm $A$, there exists an instance $C$ of $\domination$ of size $n$ such that $r_{min}(C, A) \geq \Omega\left(\frac{\sqrt{n}}{\log n}\right) r_{min}(C)$. 
\end{theorem}

\subsection{A hard distribution for domination}

To prove Theorem \ref{thm:domlb}, we will show that there exists a distribution over instances of the domination problem such that, while each instance in the support of this distribution can be solved by some algorithm with a small number of samples, any algorithm requires a large number of samples given an instance randomly sampled from this distribution.

Let $\mathcal{C}$ be a distribution over instances $C$ of the domination problem of size $n$. We extend $r_{min}$ to distributions by defining $r_{min}(\mathcal{C}, A, p)$ as the minimum number of samples algorithm $A$ needs to successfully solve $\domination$ with probability at least $p$ over instances randomly sampled from $\mathcal{C}$, and let $r_{min}(\mathcal{C}, A) = r_{min}(\mathcal{C}, A, 3/4)$. The following lemma relates the distributional sample complexity to the single instance sample complexity.

\begin{lemma}\label{lem:mix1}
For any $p > 1/2$, algorithm $A$ and any distribution $\mathcal{C}$ over instances of the domination problem, there exists a $C$ in the support of $\mathcal{C}$ such that $r_{min}(C, A, p) \geq r_{min}(\mathcal{C}, A, p)$.
\end{lemma}
\begin{proof}
Let $\eps(C, A, r)$ be the probability that algorithm $A$ errs given $r$ samples from $C$. By the definition of $r_{\min}(\mathcal{C}, A, p)$, we have that

$$\sum_{C \in \supp \mathcal{C}} \Pr_{\mathcal{C}}[C] \cdot \eps(C, A, r_{min}(\mathcal{C},A, p)) = 1-p$$

\noindent
It follows that there exists some $C^{*} \in \supp \mathcal{C}$ such that

$$\eps(C^{*}, A, r_{min}(\mathcal{C}, A, p)) \geq 1-p$$

\noindent
Since $\eps(C^{*}, A, r)$ is decreasing in $r$, this implies that $r_{min}(C^{*}, A, p) \geq r_{min}(\mathcal{C},A,p)$, as desired.
\end{proof}

We will find it useful to work with distributions that are only mostly supported on easy instances. The following lemma lets us do that.

\begin{lemma}\label{lem:mix2}
Let $\mathcal{C}$ be a distribution over instances of the domination problem, and let $E$ be an event with $\Pr[E] = 1-\delta$. Then for any algorithm $A$ and any $1-\delta > p > \frac{1}{2}$, $r_{min}(\mathcal{C}|E, A, p+\delta) \geq r_{min}(\mathcal{C}, A, p)$. 
\end{lemma}
\begin{proof}
By the definition of $r_{\min}(\mathcal{C}, A, p)$, we have that

$$\sum_{C \in \supp \mathcal{C}} \Pr_{\mathcal{C}}[C] \cdot \eps(C, A, r_{min}(\mathcal{C},A, p)) = 1-p$$

\noindent
Rewrite this as

$$\Pr[\overline{E}]\cdot\sum_{C \in \supp \mathcal{C}} \Pr_{\mathcal{C}|\overline{E}}[C] \cdot \eps(C, A, r_{min}(\mathcal{C},A, p))  + \Pr[E]\cdot\sum_{C \in \supp \mathcal{C}} \Pr_{\mathcal{C}|E}[C] \cdot \eps(C, A, r_{min}(\mathcal{C},A, p)) = 1-p$$

Since $\sum_{C \in \supp \mathcal{C}} \Pr_{\mathcal{C}|\overline{E}}[C] = 1$ and $\Pr[\overline{E}] = \delta$, it follows that

$$\sum_{C \in \supp \mathcal{C}} \Pr_{\mathcal{C}|E}[C] \cdot \eps(C, A, r_{min}(\mathcal{C},A, p)) \geq 1-p - \delta$$

\noindent
from which it follows that $r_{min}(\mathcal{C}|E, A, p+\delta) \geq r_{min}(\mathcal{C}, A, p)$.

\end{proof}

We can now define the hard distribution for the domination problem. Define $\gamma = \frac{1}{100\sqrt{n}}$. Let $S_{P}$ be a random subset of $[n]$ where each $i \in [n]$ is independently chosen to belong to $S_{P}$ with probability $\gamma$. Likewise, define $S_{Q}$ the same way (independently of $S_P$). Finally, fix $n$ constants $R_i$ all in the range $[\frac{1}{4}, \frac{3}{4}]$ (for now, it is okay to consider only the case where $R_i = \frac{1}{2}$ for all $i$; to extend this lower bound to the top-$k$ problem, we will need to choose different values of $R_i$). Then the hard distribution $\mathcal{C}_{hard}$ is the distribution over instances $C(S_P, S_Q) = (n, \bp, \bq)$ of $\domination$ where

$$ p_i = \begin{cases}R_i(1+\eps) & \mbox{ if } i \in S_{P} \\
R_i & \mbox{ if } i \not\in S_{P} \end{cases}$$

\noindent
and

$$ q_i = \begin{cases}R_i(1-\eps) & \mbox{ if } i \in S_{Q} \\
R_i & \mbox{ if } i \not\in S_{Q} \end{cases}$$

We claim that the majority of the instances in the support of $\mathcal{C}_{hard}$ have an algorithm that requires few samples. Intuitively, if $S_P$ and $S_Q$ are fixed, then the best algorithm for that specific instance can restrict attention only to the indices in $S_P$ and $S_Q$. In particular, if $S_P$ is large enough (some constant times its expected size), then simply throwing away all indices not in $S_P$ and counting which row has more heads is an efficient algorithm for recovering the dominant set.

\begin{theorem}\label{thm:hardub}
Fix any $S_P$ and $S_Q$ such that $|S_P| \geq \frac{1}{10}n\gamma$. Then $r_{min}(C(S_P, S_Q), p) = O\left(\frac{\log(1-p)^{-1}}{\eps^2\sqrt{n}}\right)$ for all $p<1$. 
\end{theorem}
\begin{proof}
It suffices to demonstrate an algorithm $A$ such that $r_{min}(C(S_P, S_Q), A, p) = O\left(\frac{\log(1-p)^{-1}}{\eps^2\sqrt{n}}\right)$. 

Any algorithm $A$ receives two sets $X, Y$, each of $r$ samples from $n$ coins. Write $X = (X_1, X_2, \dots, X_n)$, where each $X_i = (X_{i,1}, X_{i,2}, \dots X_{i,r})$ is the collection of $r$ samples from coin $i$ (likewise, write $Y = (Y_1, Y_2, \dots, Y_n)$, and $Y_i = (Y_{i,1}, Y_{i,2}, \dots Y_{i,r})$). Consider the following algorithm: $A$ computes the value

$$T = \sum_{i \in S_{P}}\sum_{j=1}^{r} (X_{i,j} - Y_{i,j})$$

\noindent
and outputs that $B=0$ if $T \geq 0$ and outputs $B=1$ otherwise. 

For each $i, j$, let $A_{i,j} = X_{i,j} - Y_{i,j}$. If $B=0$, then $A_{i,j} \in [-1, 1]$, $\E[A_{i,j}] \geq \eps R_i \geq \frac{\eps}{4}$ and all the $A_{i,j}$ are independent. It follows from Hoeffding's inequality that in this case,

\begin{eqnarray*}
\Pr[T < 0] &=& \Pr[T - \E[T] < -\E[T]] \\
&\leq & \exp\left(- \frac{2E[T]^2}{4|S_P|r}\right) \\
&=& \exp\left(-\frac{|S_P|r\eps^2}{32}\right) \\
&\leq & \exp\left(-\frac{\gamma n\eps^2r}{320}\right) \\
&=& \exp\left(-\frac{\sqrt{n}\eps^2r}{32000}\right)
\end{eqnarray*}

Therefore, choosing $r = \frac{32000\ln (1-p)^{-1}}{\sqrt{n}{\eps^2}} = O\left(\frac{\log(1-p)^{-1}}{\sqrt{n}\eps^2}\right)$ guarantees $\Pr[T < 0] \leq 1-p$. Similarly, the probability that $T \geq 0$ if $B=1$ is also at most $1-p$ for this $r$. The conclusion follows. 
\end{proof}

By a simple Chernoff bound, we also know that the event that $S_P$ has size at least $\frac{1}{10}n\gamma$ occurs with high probability.

\begin{lemma}\label{lem:chern}
$\Pr\left[|S_P|\geq \frac{1}{10}n\gamma\right] \geq 1- e^{-\sqrt{n}/400}$.
\end{lemma}

In the following subsection, we will prove that for all $A$, $r_{min}(\mathcal{C}_{hard}, A)$ is large. More precisely, we will prove the following theorem.

\begin{theorem}\label{thm:harddist}
For all algorithms $A$, $r_{min}(\mathcal{C}_{hard}, A, \frac{2}{3}) = \Omega\left(\frac{1}{\eps^2\log n}\right)$. 
\end{theorem}

Given that this theorem is true, we can complete the proof of Theorem~\ref{thm:domlb}.

\begin{proof}[Proof of Theorem~\ref{thm:domlb}]
By Theorem \ref{thm:harddist}, for any algorithm $A$, $r_{min}(\mathcal{C}_{hard}, A, \frac{2}{3}) = \Omega\left(\frac{1}{\eps^2\log n}\right)$. Let $E$ be the event that $|S_P| \geq \frac{1}{10} n\gamma$. By Lemma \ref{lem:chern}, if $n \geq (400\ln \frac{12}{11})^2$, $\Pr[E] \geq \frac{1}{12}$. It then follows from Lemma \ref{lem:mix2} that 

\begin{eqnarray*}
r_{min}(\mathcal{C}_{hard}|E, A) &=& r_{min}(\mathcal{C}_{hard}|E, A, 3/4)\\
 &\geq& r_{min}(\mathcal{C}_{hard}, A, 2/3)\\
 &\geq& \Omega\left(\frac{1}{\eps^2\log n}\right).
 \end{eqnarray*}

It then follows by Lemma \ref{lem:mix1} that there is a specific instance $C = C(S_P, S_Q)$ with $|S_{P}|$ at least $\frac{1}{10}\gamma n$ such that $r_{min}(C, A) \geq \Omega\left(\frac{1}{\eps^2\log n}\right)$. On the other hand, by Theorem \ref{thm:hardub}, for this $C$, $r_{min}(C) \leq O\left(\frac{1}{\eps^2\sqrt{n}}\right)$. It follows that for any algorithm $A$, there exists an instance $C$ such that $r_{min}(C, A) \geq \Omega\left(\frac{\sqrt{n}}{\log n}\right) r_{min}(C)$, as desired.

\end{proof}

\subsection{Proof of hardness}

In this subsection, we prove Theorem \ref{thm:harddist}; namely, we will show that any algorithm needs at least $\Omega\left(\frac{1}{\eps^2\log n}\right)$ samples to succeed on $\mathcal{C}_{hard}$ with constant probability. Our main approach will be to bound the mutual information between the samples provided to the algorithm and the correct output (recall that $B$ is the hidden bit that determines whether the samples in $X$ are drawn from $\bp$ or from $\bq$).

\begin{lemma}\label{lem:fano}
If $I(XY;B) < 0.05$, then there is no algorithm that can succeed at identifying $B$ with probability at least $\frac{2}{3}$.
\end{lemma}
\begin{proof}
Fix an algorithm $A$, and let $p_e$ be the probability that it errs at computing $B$. By Fano's inequality, we have that

\begin{eqnarray*}
H(p_e) &\geq& H(B|XY)\\
&=& H(B) - I(XY;B) \\
&=& 1 - I(XY;B)\\
&>& 0.95
\end{eqnarray*}

Since $H(\frac{1}{3}) \leq 0.95$, it follows that $A$ must err with probability at least $1/3$.
\end{proof}

Via the chain rule, we can decompose $I(XY;B)$ into the sum of many smaller mutual informations.

\begin{lemma}\label{lem:infdecomp}
$I(XY;B) \leq \sum_{i=1}^{n} \left(I(X_i;B) + I(Y_i;B)\right)$
\end{lemma}
\begin{proof}
Write $X^{<i}$ to represent the concatenation $X_1X_2\dots X_{i-1}$. By the chain rule, we have that

$$I(XY;B) = \sum_{i=1}^{n} I(X_iY_i;B|X^{<i}Y^{<i})$$

We claim that $I(X_iY_i;X^{<i}Y^{<i}|B) = 0$. To see this, note that given $B$, each coin in $X_i$ is sampled from some $\Ber(p)$ distribution, where $p$ only depends on whether $i \in S_P$ or $i\in S_Q$. Since each $i$ is chosen to belong to $S_P$ and $S_Q$ independently with probability $\gamma$, this implies $X_i$ (and similarly $Y_i$) are independent from $X^{<i}$ and $Y^{<i}$ given $B$. By Fact \ref{fact:it1}, this implies that $I(X_iY_i;B|X^{<i}Y^{<i}) \leq I(X_iY_i;B)$, and therefore that

$$I(XY;B) \leq \sum_{i=1}^{n} I(X_iY_i;B).$$

Likewise, we can write $I(X_iY_i;B) = I(X_i;B) + I(Y_i;B|X_i)$. Since $I(X_i;Y_i|B) = 0$ (since $S_P$ and $S_Q$ are chosen independently), again by Fact \ref{fact:it1} it follows that $I(Y_i;B|X_i) \leq I(Y_i;B)$ and therefore that

$$I(XY;B) \leq \sum_{i=1}^{n} \left(I(X_i;B) + I(Y_i;B)\right).$$
\end{proof}

\begin{lemma}
\label{lem:inflb}
If $n \geq 400$ and $r = \frac{1}{100\eps^2\ln n}$, then for all $i$, $I(B;X_i)=I(B;Y_i) \leq \frac{1}{100n}$. 
\end{lemma}
\begin{proof}
By symmetry, $I(B;X_i)=I(B;Y_i)$. We will show that $I(B;X_i) \leq \frac{1}{100n}$. 

Let $Z_i = \sum_{j} X_{i,j}$. Note that $Z_i$ is a sufficient statistic for $B$, and therefore $I(B;X_i) = I(B;Z_i)$. By Fact \ref{fact:div}, 

\begin{eqnarray*}
I(B;Z_i) &=& \E_{Z_{i}}[D(B|Z_{i} \| B)]\\
&=& \sum_{z=0}^{r} \Pr[Z_i=z] \cdot D(\Pr[B=0|Z_i=z] \| \tfrac{1}{2}).
\end{eqnarray*}

We next divide the range of $z$ into two cases.

\begin{enumerate}
\item 
\textbf{Case 1}: $|z-rR_i| \leq 11r\eps \ln n$.

In this case, we will bound the size of $D(\Pr[B=0|Z_i=z] \| \tfrac{1}{2})$. Note that

\begin{eqnarray}
\left|\Pr[B = 0|Z_i = z]-\frac{1}{2}\right| &=& \left|\frac{\Pr[Z_i = z|B=0] \cdot \Pr[B=0]}{\Pr[Z_i = z]} - \frac{1}{2}\right| \nonumber \\
&=& \left|\frac{\Pr[Z_i=z|B=0]}{\Pr[Z_i=z|B=0] + \Pr[Z_i=z|B=1]} - \frac{1}{2}\right| \nonumber \\
&=& \frac{\left|\Pr[Z_i=z|B=0]-\Pr[Z_i=z|B=1]\right|}{2(\Pr[Z_i=z|B=0] + \Pr[Z_i=z|B=1])}
\label{eqn:halfdist}
\end{eqnarray}

Now, note that

\begin{eqnarray*}
\Pr[Z_i=z|B=0] &=& (1-\gamma)\binom{r}{z}R_i^{z}(1-R_{i})^{r-z} + \gamma\binom{r}{z}(R_i(1+\eps))^{z}(1-R_i(1+\eps))^{r-z} \\
\Pr[Z_i=z|B=1] &=& (1-\gamma)\binom{r}{z}R_i^{z}(1-R_{i})^{r-z} + \gamma\binom{r}{z}(R_i(1-\eps))^{z}(1-R_i(1-\eps))^{r-z} \\
\end{eqnarray*}

We can therefore lower bound the denominator of (\ref{eqn:halfdist}) via

\begin{eqnarray*}
2(\Pr[Z_i=z|B=0] + \Pr[Z_i=z|B=1]) &\geq& 4(1-\gamma)\binom{r}{z}R_i^{z}(1-R_{i})^{r-z} \\
&\geq& 2\binom{r}{z}R_i^{z}(1-R_{i})^{r-z}
\end{eqnarray*}

Likewise, we can write the numerator of (\ref{eqn:halfdist}) as

\begin{equation*}
\left|\Pr[Z_i=z|B=0]-\Pr[Z_i=z|B=1]\right|=\gamma\binom{r}{z}R_i^{z}(1-R_i)^{r-z}M
\end{equation*}

\noindent
where

\begin{eqnarray*}
M &=& \left|(1+\eps)^{z}\left(\frac{1-R_i(1+\eps)}{1-R_i}\right)^{r-z} - (1-\eps)^{z}\left(\frac{1-R_i(1-\eps)}{1-R_i}\right)^{r-z}\right|\\
&=& \left|(1+\eps)^{z}\left(1 - \frac{R_i}{1-R_i}\eps\right)^{r-z} - (1-\eps)^{z}\left(1+\frac{R_i}{1-R_i}\eps\right)^{r-z}\right|.
\end{eqnarray*}

To bound $M$, note that (applying the inequality $1+x \leq e^{x}$)

\begin{eqnarray*}
(1+\eps)^{z}\left(1 - \frac{R_i}{1-R_i}\eps\right)^{r-z} &\leq& \exp\left(\eps z - \eps\frac{R_i}{1-R_i}(r-z)\right) \\
&=& \exp\left(\eps \frac{z-rR_i}{1-R_i}\right) \\
&\leq& \exp(4\eps(z-rR_i)) \\
&\leq& \exp(44r\eps^2\ln n) \\
&=& e^{0.44} \\
&<& 2
\end{eqnarray*}

Similarly, $(1-\eps)^{z}\left(1+\frac{R_i}{1-R_i}\eps\right)^{r-z} \leq 2$. It follows that $M \leq 2$, and therefore that
 
\begin{eqnarray*}
\left|\Pr[B = 0|Z_i = z] - \frac{1}{2}\right| &=& \frac{\left|\Pr[Z_i=z|B=0]-\Pr[Z_i=z|B=1]\right|}{2(\Pr[Z_i=z|B=0] + \Pr[Z_i=z|B=1])}\\
&\leq & \frac{\gamma\binom{r}{z}R_i^{z}(1-R_i)^{r-z}M}{2\binom{r}{z}R_i^{z}(1-R_i)^{r-z}} \\
&=& \frac{\gamma M}{2}\\
&\leq & \gamma
\end{eqnarray*}

By Fact \ref{fact:repinsker}, this implies that

$$D(\Pr[B=0|Z_i=z] \| \tfrac{1}{2}) \leq \frac{4\gamma^2}{\ln 2}.$$

\item
\textbf{Case 2}: $|z - rR_i|  > 11 r\eps\ln n $.

Let $Z^{+}$ be the sum of $r$ i.i.d. $\Ber\left(R_i(1+\eps)\right)$ random variables. Note that since $Z$ is the sum of $r$ $\Ber(p)$ random variables for some $p \leq R_i(1+\eps)$, $\Pr[Z^{+} \geq x] \geq \Pr[Z \geq x]$ for all $x$. Therefore, by Hoeffding's inequality, we have that

\begin{eqnarray*}
\Pr\left[Z-rR_i \geq 11r \eps\ln n\right] &\leq& \Pr\left[Z^{+}-rR_i \geq 11r\eps\ln n \right] \\
&\leq& \Pr\left[Z^{+} - rR_i(1+\eps) \geq r\eps (11\ln n - R_i)\right] \\
&\leq& \Pr\left[Z^{+} - \E[Z^{+}] \geq 10r\eps\ln n\right] \\
&\leq& \exp\left(-\frac{2(10r\eps \ln n)^2}{r}\right) \\
&=& \exp(-2\ln n) \\
&=& n^{-2}
\end{eqnarray*}

Likewise, we can show that

$$\Pr\left[Z-rR_i \leq -11r \eps\ln n\right] \leq n^{-2}$$

\noindent
so

$$\Pr\left[\left|Z-rR_i\right| \geq 11r \eps\ln n\right] \leq 2n^{-2}$$
\end{enumerate}

Combining these two cases, we have that (for $n \geq 400$)
\begin{eqnarray*}
I(B;Z_i) &=& \sum_{z=0}^{r} \Pr[Z_i=z] \cdot D(\Pr[B=0|Z_i=z] \| \tfrac{1}{2}) \\
 &\leq& \sum_{|\|z\| - r/2|  > 11r\eps\ln n } \Pr[Z_i = z] \cdot 1  +  \sum_{|\|z\| - r/2|  \leq 11r\eps\ln n} \Pr[Z_i = z] \cdot O(\gamma^2) \\
 &\leq& 2n^{-2} + \frac{4\gamma^2}{\ln 2} \\
 &\leq& \frac{1}{100n}.
\end{eqnarray*}

\end{proof}

We can now complete the proof of Theorem \ref{thm:harddist}.

\begin{proof}[Proof of Theorem~\ref{thm:harddist}]
Combining Lemmas \ref{lem:infdecomp} and \ref{lem:inflb}, we have that if $r = \frac{1}{100\eps^2\ln n}$, then (for $n \geq 400$) $I(XY;B) \leq 2n I(X_i;B) \leq 0.02$. Therefore by Lemma \ref{lem:fano}, there exists no algorithm $A$ that, given this number of samples, correctly identifies $B$ (and thus solves the domination problem) with probability at least $2/3$. It follows that 

$$r_{min}(\mathcal{C}_{hard}, A, \tfrac{2}{3}) \geq \frac{1}{100\eps^2\ln n} = \Omega\left(\frac{1}{\eps^2\log n}\right)$$

\noindent
as desired.

\end{proof}

\subsection{Proving hardness for Top-$K$}

We will now show how to use our hard distribution of instances of $\domination$ to generate a hard distribution of instances of $\topk$. Our goal will be to embed our $\domination$ instance as rows $k$ and $k+1$ of our SST matrix; hence, intuitively, deciding which of the two rows ($k$ or $k+1$) belongs to the top $k$ is as hard as solving the domination problem.

Unfortunately, the SST condition imposes additional structure that prevents us from directly embedding any instance of the domination problem. However, for appropriate choices of the constants $R_i$, all instances in the support of $\mathcal{C}_{hard}$ give rise to valid SST matrices. 

Specifically, we construct the following distribution $\mathcal{S}_{hard}$ over $\topk$ instances $S$ of size $n+2$. Consider the distribution $\mathcal{C}_{hard}$ over $\domination$ instances of size $n$, where for $1\leq i \leq n$, $R_i = \frac{1}{4} + \frac{i}{8n}$, and $\eps = \frac{1}{100n^2}$. Now, consider the following map $f$ from $\domination$ instances $C = (\bp, \bq)$ to $\topk$ instances $S = f(C) = (n+2, k, \bP)$: we choose $k=n+1$ (so that the problem becomes equivalent to identifying row $n+2$) and define the matrix $\bP$ as follows:

$$\bP_{ij} = \begin{cases}
\bp_j &\mbox{if } i=n+1 \mbox{ and } j \leq n\\
\bq_j &\mbox{if } i=n+2 \mbox{ and } j \leq n\\
1-\bp_i &\mbox{if } j=n+1 \mbox{ and } i \leq n\\
1-\bq_i &\mbox{if } j=n+2 \mbox{ and } i \leq n\\ 
\frac{1}{2} &\mbox{otherwise}
\end{cases}$$

In general, for arbitrary $\bp$ and $\bq$, this matrix may not be an SST matrix. Note however that for this choice of $R_i$ and $\eps$, it is always the case that $R_i(1+\eps) \leq R_{i+1}(1-\eps)$, so for all $i$ (regardless of sample $C$), $p_i < p_{i+1}$. In addition, all the $R_i$ belong to $[1/4, 3/8]$, so for all $i$, $p_i$ and $q_i$ are less than $1/2$. From these two observations, it easily follows that if $C$ belongs to the support of $\mathcal{C}_{hard}$, $\bP$ is an SST matrix, and $f(C)$ is a valid instance of the top-$k$ problem. We will write $\mathcal{S}_{hard} = f(\mathcal{C}_{hard})$ to denote the distribution of instances of top-$k$ $f(C)$ where $C$ is sampled from $\mathcal{C}_{hard}$. Likewise, for any event $E$ (e.g. the event that $|S_{P}| \geq \frac{1}{10}n\gamma$) , we write $\mathcal{S}_{hard}|E$ to denote the distribution $f(\mathcal{C}_{hard}|E)$. 

We will begin by showing that, if there exists a sample efficient algorithm for some $\domination$ instance $C$ in the support of $\mathcal{C}_{hard}$, there exists a similarly efficient algorithm for the corresponding $\topk$ instance $S=f(C)$.

\begin{lemma}\label{lem:sstdownred}
If $C \in \supp \mathcal{C}_{hard}$ and $S = f(C)$, then $r_{min}(S) \leq \max(r_{min}(C, \frac{4}{5}), 1000n^2(1+\ln n))$. 
\end{lemma}
\begin{proof}
Let $A$ be an algorithm that successfully solves the $\domination$ instance $C$ with probability at least $\frac{4}{5}$ using $r_{min}(C, \frac{4}{5})$ samples. We will show how to use $A$ to construct an algorithm $A'$ that solves the $\topk$ instance $S$ with probability at least $3/4$ using $r = \max(r_{min}(C, \frac{4}{5}), 1000n^2(1+\ln n))$ samples. 

For each $i, j$, write $Z_{i, j} = \sum_{\ell = 1}^{r}Z_{i, j, \ell}$. Our algorithm $A'$ operates as follows. 

\begin{enumerate}
\item We begin by finding the two rows with the smallest row sums $\sum_{j}Z_{i, j}$. Let these two rows have indices $c$ and $d$. We claim that, with high probability, $\pi^{-1}(\{c, d\}) = \{n+1, n+2\}$. 

To see this, note that for all $i \not\in \pi(\{n+1, n+2\})$, $\bP_{i, j} \geq \frac{1}{2}$, so $\E\left[\sum_{j}Z_{i,j}\right] \geq \left(\frac{n}{2} + 1\right)r$. Thus, for any fixed $i \not\in \pi(\{n+1, n+2\})$, it follows from Hoeffding's inequality that

$$\Pr\left[\sum_{j}Z_{i,j} \leq \left(\frac{7}{16}n+1\right)r\right] \leq \exp\left(-\frac{nr}{128}\right)$$

\noindent
so by the union bound, the probability that there exists an $i \not\in \pi^{-1}(\{n+1, n+2\})$ such that $\sum_{j}Z_{i,j} \leq \left(\frac{7}{16}n+1\right)r$ is at most $n\exp\left(-\frac{nr}{128}\right)$. 

On the other hand, if $i \in \pi(\{n+1, n+2\})$ then $\bP_{i, j} \leq \frac{3}{8}(1+\eps)$ unless $j \in \pi(\{n+1, n+2\})$, where $\bP_{i,j} = \frac{1}{2}$; it follows that in this case, $\E\left[\sum_{j}Z_{i, j}\right] \leq \left(\frac{3n}{8}(1+\eps) + 1 \right)r$. Similarly, applying Hoeffding's inequality in this case, we find that for any fixed $i \in \pi^{-1}(\{n+1, n+2\})$,

$$\Pr\left[\sum_{j}Z_{i,j} \geq \left(\frac{7}{16}n+1\right)r\right] \leq \exp\left(-\frac{nr}{128(1+\eps)^2}\right) \leq 1.5\exp\left(-\frac{nr}{128}\right)$$

\noindent
and thus the probability that there exists some $i\in \pi^{-1}(\{n+1, n+2\})$, such that $\sum_{j}Z_{i,j} \geq \left(\frac{7}{16}n+1\right)r$ is at most $3\exp\left(-\frac{nr}{128}\right)$. It follows that, altogether, the probability that $\pi^{-1}(\{c, d\}) \neq \{n+1, n+2\}$ is at most $(n+3)\exp\left(-\frac{nr}{128}\right)$. Since $r \geq 1000n^2\ln n$, this is at most $4\exp(-1000/128) < 0.01$. 

\item We next sort the values $Z_{c, j}$ for $j \in [n+2] \setminus \{c, d\}$ and obtain indices $j_1, j_2, \dots, j_n$ so that $Z_{c,j_1} \leq Z_{c, j_2} \leq \dots \leq Z_{c, j_n}$. We claim that, with high probability, for all $a$, $\pi^{-1}(j_{a}) = a$.

For each $i$, let $U_i$ be the interval $\left[R_{i}(1-\eps) - \frac{1}{20n}, R_{i}(1+\eps) + \frac{1}{20n}\right]$. Note that, by our choice of $R_i$ and $\eps$, all the intervals $U_i$ are disjoint, with $U_i$ less than $U_{i+1}$ for all $i$. We will show that with high probability, $\frac{1}{r}Z_{c, \pi(i)} \in U_i$ for all $i$, thus implying the previous claim.

Note that $Z_{c, \pi(i)}$ is the sum of $r$ $\Ber(p)$ random variables, where $p$ is either $(1+\eps)R_i$, $R_i$, or $(1-\eps)R_i$. By Hoeffding's inequality, it follows that

\begin{eqnarray*}
\Pr\left[Z_{c, \pi(i)} \geq r\left(R_i(1+\eps) + \frac{1}{20n}\right)\right] &\leq& \exp\left(-2\frac{(r/20n)^2}{r}\right)\\
&=& \exp\left(-\frac{r}{200n^2}\right)
\end{eqnarray*}

Likewise,

$$\Pr\left[Z_{c, \pi(i)} \leq r\left(R_i(1-\eps) - \frac{1}{20n}\right)\right] \leq \exp\left(-\frac{r}{200n^2}\right)$$

Thus, for any fixed $i$,

$$\Pr\left[\frac{Z_{c, \pi^(i)}}{r} \not\in U_{i}\right] \leq 2\exp\left(-\frac{r}{200n^2}\right)$$

\noindent
and by the union bound, the probability this fails for some $i$ is at most $2n\exp\left(-\frac{r}{200n^2}\right)$. Since $r \geq 1000n^2(1+\ln n)$, $\exp\left(-\frac{r}{200n^2}\right) \leq (ne)^{-5}$, so this probability is at most $2e^{-5} < 0.02$.

\item Finally, we give algorithm $A$ as input $X_{i, \ell} = Z_{c, j_{i}, \ell}$ and $Y_{i, \ell} = Z_{d, j_{i}, \ell}$. Note that (conditioned on the above two claims holding), this input is distributed equivalently to input from the $\domination$ instance $C$. In particular, if $\pi^{-1}(c) = n+1$ and $\pi^{-1}(d) = n+2$, then each $X_{i, \ell}$ is distributed according to $\Ber(\bp_i)$ and each $Y_{i, \ell}$ is distributed according to $\Ber(\bq_i)$, and if $\pi^{-1}(c) = n+2$ and $\pi^{-1}(d) = n+1$, then each $X_{i, \ell}$ is distributed according to $\Ber(\bq_{i})$ and each $Y_{i, \ell}$ is distributed according to $\Ber(\bp_{i})$. Thus, if $A$ returns $B=0$, we return $[n+2] \setminus \{d\}$ as the top $n+1$ indices, and if $A$ returns $B=1$, we return $[n+2] \setminus \{c\}$ as the top $n+1$ indices.

The probability that $A$ fails given that steps 1 and 2 succeed is at most $0.2$, and the probability that either of the two steps fail to succeed is at most $0.01+0.02 = 0.03$. Since  $0.2 + 0.03 < \frac{1}{4}$, $A'$ succeeds with probability at least $\frac{3}{4}$, as desired.
\end{enumerate}

\end{proof}

\begin{corollary}\label{cor:sstub}
Let $E$ be the event that $|S_{P}| \geq \frac{1}{10}n\gamma$. If $C \in \supp (\mathcal{C}_{hard}|E)$ and $S = f(C)$, then $r_{min}(S) \leq O(n^{3.5})$.
\end{corollary}
\begin{proof}
Recall that by Theorem \ref{thm:hardub}, for any $C \in \supp(\mathcal{C}_{hard}|E)$, $r_{min}(C,\frac{4}{5}) \leq O\left(\frac{1}{\sqrt{n}\eps^2}\right) = O(n^{3.5})$. By Lemma \ref{lem:sstdownred}, $r_{min}(S) \leq \max(r_{min}(C, \frac{4}{5}), 1000n^2(1+\ln n)) \leq O(n^{3.5})$.
\end{proof}

We next show that solving $\topk$ over the distribution $\mathcal{S}_{hard}|E$ is at least as hard as solving $\domination$ over the distribution $\mathcal{C}_{hard}|E$.

\begin{lemma}\label{lem:sstupred}
For any algorithm $A$ that solves $\topk$, there exists an algorithm $A'$ that solves domination such that $r_{min}(\mathcal{S}_{hard}, A, p) \geq \frac{1}{2}r_{min}(\mathcal{C}_{hard}, A', p)$.
\end{lemma}
\begin{proof}
We will show more generally that for any distribution $\mathcal{C}$ of $\domination$ instances, if $\mathcal{S} = f(\mathcal{C})$ is a valid distribution of $\topk$ instances, then $r_{min}(\mathcal{S}, A, p) \geq \frac{1}{2}r_{min}(\mathcal{C}, A', p)$. 

We will construct $A'$ by embedding the domination instance inside a top-$k$ instance in much the same way that the function $f$ does, and then using $A$ to solve the top-$k$ instance. We receive as input two sets of samples $X_{i,\ell}$ and $Y_{i, \ell}$ (where $1\leq i,j \leq n$ and $1\leq \ell \leq r$) from some $\domination$ instance $C$ drawn from $\mathcal{C}$. We then generate a random permutation $\pi$ of $[n+2]$. We use our input and this permutation to generate a matrix $Z_{i, j, \ell}$ (where $1\leq i, j \leq n+2$ and $1\leq \ell \leq \frac{r}{2}$) of samples to input to $A$ as follows.

For $1 \leq i, j \leq n$, set each $Z_{\pi(i), \pi(j), \ell}$ to be a random $\Ber(\frac{1}{2})$ random variable. Similarly, for $n+1 \leq i, j \leq n+2$, set each $Z_{\pi(i), \pi(j), \ell}$ to be a random $\Ber(\frac{1}{2})$ random variable. Now, for all $1 \leq j \leq n$, set $Z_{\pi(n+1), \pi(j), \ell} = X_{j, \ell}$ and set $Z_{\pi(n+2), \pi(j), \ell} = Y_{j, \ell}$. Similarly, for all $1 \leq i \leq n$, set $Z_{\pi(i), \pi(n+1), \ell} = 1-X_{i, \ell + r/2}$ and set $Z_{\pi(i), \pi(n+2), \ell} = 1-Y_{i, \ell + r/2}$. Finally, set $k=n+1$ and ask $A$ to solve the $\topk$ instance defined by $k$ and $Z_{i,j,\ell}$. If $A$ returns that $\pi(n+1)$ is in the top $n+1$ indices, return $B=0$, and otherwise return $B=1$.

From our construction, if the $r$ samples of $X$ and $Y$ are distributed according to a $\domination$ instance $C$, then the $r/2$ samples of $Z$ are distributed according to the $\topk$ instance $S=f(C)$. Since $A$ succeeds with probability $p$ on distribution $\mathcal{S}$ with $r_{min}(\mathcal{S}, A, p)$ samples, $A'$ therefore succeeds with probability $p$ on distribution $\mathcal{C}$ with $2r_{min}(\mathcal{S}, A, p)$ samples, thus implying that $r_{min}(\mathcal{S}, A, p) \geq \frac{1}{2}r_{min}(\mathcal{C}, A', p)$.
\end{proof}
\begin{corollary}\label{cor:sstlb}
For all algorithms $A$ that solve $\topk$, $r_{min}(\mathcal{S}_{hard}, A, \frac{2}{3}) = \Omega\left(\frac{n^4}{\log n}\right)$. 
\end{corollary}
\begin{proof}
Theorem \ref{thm:harddist} tells us that for all algorithms $A'$ that solve $\domination$, $r_{min}(\mathcal{C}_{hard}, A, \frac{2}{3}) = \Omega\left(\frac{1}{\eps^2\log n}\right) = \Omega\left(\frac{n^4}{\log n}\right)$. Combining this with Lemma \ref{lem:sstupred}, we obtain the desired result.
\end{proof}

We can now prove Theorem \ref{thm:sstlb} in much the same fashion as Theorem \ref{thm:domlb}.

\begin{proof}[Proof of Theorem~\ref{thm:sstlb}]
By Corollary \ref{cor:sstlb}, $r_{min}(\mathcal{S}_{hard}, A, \frac{2}{3}) = \Omega\left(\frac{n^4}{\log n}\right)$. Let $E$ be the event that $|S_P| \geq \frac{1}{10}n\gamma$ (in the original $\domination$ instance $C$). By Lemma \ref{lem:chern}, if $n \geq (400 \ln \frac{12}{11})^2$, $\Pr[E] \geq \frac{1}{12}$, and it follows from Lemma \ref{lem:mix2} that 

\begin{eqnarray*}
r_{min}(\mathcal{S}_{hard}|E, A) &=& r_{min}(\mathcal{S}_{hard}|E, A, \frac{3}{4}) \\
&\geq & r_{min}(\mathcal{S}_{hard}, A, \frac{2}{3}) \\
&\geq & \Omega\left(\frac{n^4}{\log n}\right)
\end{eqnarray*}

It therefore follows from \ref{lem:mix1} that there is a specific instance $S$ in the support of $\mathcal{S}_{hard}|E$ such that $r_{min}(S, A) \geq \Omega\left(\frac{n^4}{\log n}\right)$. However, by Corollary \ref{cor:sstub}, $r_{min}(S) \leq O(n^{3.5})$. It follows that for any algorithm $A$, there exists an instance $S$ of $\topk$ such that $r_{min}(S, A) \geq \Omega\left(\frac{\sqrt{n}}{\log n}\right)r_{min}(S)$, as desired.
\end{proof}

%-----------------------------------------------------------------------------------------------------------------------------

%\section{Experiments}\label{sec:experiments}
%\input{experiments}

%-----------------------------------------------------------------------------------------------------------------------------

%\bibliographystyle{abbrvnat}
\bibliographystyle{alpha}
\bibliography{references}

\appendix
\section{Probability and Information Theory Preliminaries}

\label{app:info}

%!TEX root = main.tex

We briefly review some standard facts and definitions from information theory we will use throughout this paper. For a more detailed introduction, we refer the reader to \cite{CK11}.

Throughout this paper, we use $\log$ to refer to the base $2$ logarithm and use $\ln$ to refer to the natural logarithm. 

\begin{definition}
The \emph{entropy} of a random variable $X$, denoted by $H(X)$, is defined as $H(X) = \sum_x \Pr[X = x] \log(1 / \Pr[X = x])$. 
\end{definition}

If $X$ is drawn from Bernoulli distributions $\Ber(p)$, we use $H(p) = -(p\log p + (1-p)(\log(1-p))$ to denote $H(X)$. 

\begin{definition}
The \emph{conditional entropy} of random variable $X$ conditioned on random variable $Y$ is defined as $H(X|Y) = \mathbb{E}_y[H(X|Y = y)]$. 
\end{definition}

\begin{fact}
$H(XY) = H(X) + H(Y|X)$. 
\end{fact}

\begin{definition}
The \emph{mutual information} between two random variables $X$ and $Y$ is defined as $I(X;Y) = H(X) - H(X|Y) = H(Y) - H(Y|X)$. 
\end{definition}

\begin{definition}
The \emph{conditional mutual information} between $X$ and $Y$ given $Z$ is defined as $I(X;Y|Z) = H(X|Z) - H(X|YZ) = H(Y|Z) - H(Y|XZ)$. 
\end{definition}

\begin{fact}\label{fact:cr}
Let $X_1,X_2,Y,Z$ be random variables, we have $I(X_1X_2;Y|Z) = I(X_1;Y|Z) + I(X_2;Y|X_1Z)$.
\end{fact}

\begin{fact}
\label{fact:it1}
Let $X,Y,Z,W$ be random variables. If $I(Y;W|X,Z) = 0$, then $I(X;Y|Z) \geq I(X;Y|ZW)$. 
\end{fact}

\begin{definition}
The \emph{Kullback-Leibler divergence} between two random variables $X$ and $Y$ is defined as $D(X\| Y) = \sum_x \Pr[X = x] \log(\Pr[X = x] / \Pr[Y = x])$. 
\end{definition}

If $X$ and $Y$ are drawn from Bernoulli distribution $B_p$ and $B_q$, we write $D(p\| q)$ as an abbreviation for $D(X\| Y)$. 

\begin{fact}
\label{fact:div}
Let $X,Y,Z$ be random variables, we have $I(X;Y|Z) = \mathbb{E}_{x,z}[D((Y|X = x, Z=z)\|(Y|Z=z))]$.
\end{fact}

\begin{fact}\label{fact:repinsker}
Let $X,Y$ be random variables, 
\[
\sum_x \frac{|\Pr[X = x] - \Pr[Y = x]|^2}{2\max\{\Pr[X = x], \Pr[Y=x]\}} \leq \ln(2) \cdot D(X\|Y)  \leq \sum_x \frac{|\Pr[X = x] - \Pr[Y = x]|^2}{\Pr[Y=x]}.
\]
\end{fact}
\begin{proof}
A proof of Fact \ref{fact:repinsker} can be found in \cite{BM15}.
\end{proof}

We will also need the following quantitative version of the central limit theorem.
\begin{lemma}[Berry-Esseen Theorem]
\label{lem:berryesseen}
Let $Z_1,\cdots,Z_k$ be independent random variables and let $S=\sum_{i=1}^k Z_i$. Let $\mu=\E[S]=\sum_{i=1}^k \E[Z_i]$, $\sigma^2=\var[S]=\sum_{i=1}^k \var[Z_i]$ and $\gamma=\sum_{i=1}^k \E[|Z_i-\E[Z_i]|^3]$. Let $\Phi$ be the CDF of standard Gaussian. Then for all $t\in \mathbb{R}$,
$$\left|\Pr\left[S< t\right]-\Phi\left(\frac{t-\mu}{\sigma}\right)\right|\le \frac{\gamma}{\sigma^3} .$$
\end{lemma}
Finally, we will need the following estimates on the tails of the Gaussian distribution.
\begin{lemma}
\label{fact:gaussiancdf}
Let $\Phi(t)$ be the CDF of standard Gaussian distribution then for $t>0$,
$$\frac{1}{\sqrt{2\pi}}\exp(-t^2/2)\left(\frac{1}{t}-\frac{1}{t^3}\right)\le 1-\Phi(t) \le \frac{1}{\sqrt{2\pi}}\exp(-t^2/2)\frac{1}{t}.$$
\end{lemma}
\begin{proof}
\begin{align*}
1-\Phi(t)&=\frac{1}{\sqrt{2\pi}}\int_t^\infty \exp(-x^2/2) dx\\
&=\frac{1}{\sqrt{2\pi}}\int_t^\infty \frac{1}{x} \cdot x\exp(-x^2/2) dx\\
&=\frac{1}{\sqrt{2\pi}}\left(\frac{\exp(-t^2/2)}{t}-\int_t^\infty \frac{1}{x^2}\exp(-x^2/2) dx\right) \tag{integration by parts}\\
&=\frac{1}{\sqrt{2\pi}}\left(\frac{\exp(-t^2/2)}{t}-\frac{\exp(-t^2/2)}{t^3}+\int_t^\infty \frac{3}{x^4}\exp(-x^2/2) dx\right)\tag{integration by parts again}.
\end{align*}
From the last two expressions, we get the required upper and lower bounds.
\end{proof}

\end{document}